\documentclass[11pt]{article}

\input{setup.sty}

\title{Holographic Entropy Inequalities and Multipartite Entanglement}

\author[a]{Sergio Hern\'andez-Cuenca,}
\emailAdd{sergiohc@mit.edu}
 
\author[b]{Veronika E. Hubeny,}
\emailAdd{veronika@physics.ucdavis.edu}

\author[b]{and Hewei Frederic Jia}
\emailAdd{fjia@ucdavis.edu}
 
\affiliation[a]{Center for Theoretical Physics, Massachusetts Institute of Technology, Cambridge, MA 02139, USA}
\affiliation[b]{Center for Quantum Mathematics and Physics (QMAP)\\ 
Department of Physics \& Astronomy, University of California, Davis, CA 95616, USA}

\abstract{
We study holographic entropy inequalities and their structural properties by making use of a judicious grouping of terms into certain multipartite information quantities. This allows us to recast cumbersome entropic expressions into much simpler ones which share interestingly rigid structures. By performing a systematic search over some of these structures, we are able to discover more than 1800 novel entropy inequalities for six parties, thereby demonstrating that these recastings provide a fruitful generating technique for uncovering new holographic entropy inequalities. In attempting to interpret the corresponding sign-definite quantities as correlation measures, we also obtain a no-go result: the superbalance property of holographic entropy inequalities turns out to preclude them from being monotonic under partial tracing. 
In the process, we also comment on the geometrical significance of multipartite information quantities and present various structural relations amongst them.
}

\begin{document}
 
 
\maketitle

\section{Introduction}
\label{sec:intro}

Holographic entropy inequalities (HEIs) restrict the entanglement structure of geometric states of any holographic CFT, beyond what would be allowed for a generic quantum state.  As such, HEIs provide some indirect insight into the emergence of spacetime and dynamical gravity, as well as the underlying workings of the holographic dictionary.  Although presently we do not have the full set of HEIs for $\nN\geq6$ parties, in this work we have developed powerful heuristics which have allowed us to collect a set of $1877$ HEI orbits for $\nN=6$.
However, when written out in terms of subsystem entropies, these quantities appear rather lengthy and uninformative;\footnote{ \, 
Other renditions, in the so-called I-basis (based on multipartite information quantities $I_n$) and K-basis (based on even-party perfect tensors) have been explored in \cite{He:2019ttu}.  Although the I and K basis representations are more compact and have a number of advantages, they only reduce the number of terms appearing in the expressions by roughly a factor of two.
} this is the case already for $\nN=5$, cf.\ \cref{tab:HEI_N5}.  Just by looking at these expressions, it is completely unclear what message they convey or what is the underlying principle they emerge from.

Recently the viewpoint that part of the complexity stems from fixing a definite value of $\nN$ was suggested by \cite{Hernandez-Cuenca:2022pst}, which showed that (subject to certain assumptions)  these HEIs can be extracted from the solution to a much simpler problem (namely a part of the ``holographic marginal independence problem" \cite{Hernandez-Cuenca:2019jpv} of characterizing all the holographically admissible sets of simultaneously decorrelated subsystems) for a more refined partition $\nN' \ge \nN$.  
However, while conceptually appealing, the underlying construction utilized extreme rays of entropy cones,  whereas the question of interest here concerns the facets of such cones.
Since the passage between the representation of a polyhedral cone in terms of its extreme rays and the representation in terms of its facets is computationally hard, the insight of \cite{Hernandez-Cuenca:2022pst} does not directly address the initial question of what the individual $\nN$-party HEIs mean physically for a given fixed $\nN$.
Correspondingly, it would be desirable to obtain a re-packaging which elucidates their physical meaning in a more manifest way.  This motivated the present exploration, in which we do present a more compact form, but more importantly a no-go theorem for a certain natural interpretation.

The program of elucidating HEIs has proceeded primarily by examining the structural properties of the holographic entropy cone (HEC), first defined in \cite{Bao:2015bfa}
and subsequently explored from several different angles \cite{Marolf:2017shp,Rota:2017ubr,Cui:2018dyq,Hubeny:2018trv,Bao:2018wwd,Hubeny:2018ijt,Cuenca:2019uzx,Czech:2019lps,Hernandez-Cuenca:2019jpv,He:2020xuo,Bao:2020zgx,Bao:2020mqq,Akers:2021lms,Avis:2021xnz,Bao:2021gzu,Czech:2022fzb,He:2022bmi,He:2023cco,He:2023rox}.
The HEC can be naturally described as a convex hull of the most extremal holographic entropy vectors, namely the \textit{extreme rays}, or equivalently in terms of the tightest HEIs that bound it,  namely the \emph{facets}. Inequalities which can be obtained as conical combinations of other HEIs are redundant and uninteresting. Facet-defining HEIs are a special set of non-redundant ones. To distinguish them one should consider the dimensionality of the space of saturating entropy vectors. More explicitly, any linear entropy inequality $Q\ge0$ specifies a half-space bounded by the hyperplane $Q=0$. A HEI defines a facet of the HEC if and only if there exists a codimension-$1$ set of linearly independent holographic entropy vectors for which $Q=0$. 
The facets of the HEC thus single out a special set of such sign-definite objects in holography which is important to understand further. 

From here on, we reserve the term HEI to indicate this facet-defining set of holographic inequalities, and we will refer to the associated sign-definite, non-redundant linear combinations of entropies as
\textit{information quantities}. Entropy vectors in the codimension-$1$ interior of any one facet will give strictly positive values for any other such information quantity. Because these information quantities are sign-definite and can all vanish independently of the others, it is tempting to view them as characterizing some sort of multipartite measures of correlations in holographic systems.

This expectation is of course borne out for a small but crucial subclass of HEIs, namely the  mutual information  quantities $\mxI{X:Y}$ measuring the total amount of correlation between disjoint subsystems $\x{X}$ and $\x{Y}$. The non-negativity of mutual information is the statement of subadditivity (SA),
\begin{equation}\label{eq:SA}
\mxI{X:Y} \equiv \Sx{X} +  \Sx{Y} - \Sx{XY} \ge 0 \ ,
\end{equation}	
where $\Sx{X}$ denotes the entropy of $\x{X}$ (and similarly for $\x{Y}$ and $\x{XY} := \x{X} \cup \x{Y}$). Structurally, the vanishing of the mutual information is indicative of the marginal independence property, that the common density matrix $\rhox{XY}$ factorizes between the two subsystems, i.e.\ the marginals $\rhox{X}$ and $\rhox{Y}$.  The fact that the mutual information is a correlation measure not only justifies its non-negativity, but also the fact that this sign-definiteness is universal for any physical state.  

Another universal statement, following from the mutual information being a correlation measure, is its monotonicity under inclusion.  This property is known as strong subadditivity (SSA), and can be expressed as the non-negativity of conditional mutual information,
\begin{equation}\label{eq:SSA}
\cmxI{X:Z}{Y} \equiv \mxI{X:YZ} - \mxI{X:Y} \ge 0 \ .
\end{equation}	
Structurally, SSA saturation is associated with a quantum Markov chain \cite{Hayden:2004}, which plays a crucial role in quantum error correction. 
However, holographically, SSA is redundant, being superseded by the monogamy of mutual information (MMI), which can be expressed as the negativity of tripartite information, 
$\txI{X:Y:Z}$  
\begin{equation}\label{eq:MMI}
-\txI{X:Y:Z} \equiv \mxI{X:YZ} - \mxI{X:Y} - \mxI{X:Z} 
\ge 0 \ .
\end{equation}	
While MMI can be violated by certain quantum states, it is satisfied by all geometric states in holography \cite{Hayden:2011ag}, and thereby offers an intriguing insight about the entanglement structure of such states.  

These simple classes of HEC facets prompt us to seek a similarly evocative explanation of the other facets,  ideally as a correlation measure or perhaps as some conditional quantity to which one could give an operational interpretation.  However, this quest is hindered by the structural complexity of the HEC growing rather quickly with the number of parties $\nN$.  
The full HEC is symmetric under any relabeling of the parties including the purifier (referred to as permutations if they do not involve the purifier, and purifications if they do), so the set of facets can be organized into symmetry orbits of a given representative instance of these HEIs.  In absence of any symmetry structure for a given HEI, the number of instances of a given orbit grows factorially with $\nN$, and the number of facet orbits also grows rapidly: for $\nN=2,3,4,5$ there are $1,2,3,8$ orbits respectively, while for $\nN=6$ our search has yielded as many as 1877 facet orbits (including $11$ distinct lifts from $\nN\le 5$), which gives a likely-modest lower bound on the actual number.  Although a systematic constructive procedure for determining the full HEC by an iterative algorithm has been formulated in \cite{Avis:2021xnz}, the more fruitful method that has hitherto yielded all of these new HEIs has involved positing a structurally compact form for generating good candidate inequalities that we present in this paper.
For each such candidate one then uses the contraction map technique developed in \cite{Bao:2015bfa} to check its validity, and then generates a suitable set of holographic entropy vectors to check if it is a facet.

To guess an obliging form for each HEI which might bring us closer to interpreting its information-theoretic meaning, we draw inspiration from the simple HEIs, namely the non-negativity of $\mxI{X:Y}$, $\cmxI{X:Y}{Z}$, and $-\txI{X:Y:Z}$. In some sense, the next simplest form would be the conditional tripartite information $-\ctxI{X:Y:Z}{W}$.  Although as shown below, this quantity by itself is not sign-definite, it turns out that many HEIs do admit a rather compact form when decomposed into a combination of (conditional) tripartite information quantities for various composite subsystems, cf.\ \cref{tab:HEI_N5} for $\nN=5$ and \cref{tab:HEI_N6} for $\nN=6$.
(The latter is just a small subset of known $\nN=6$ HEIs; the full set of hitherto-obtained ones is explained in \cref{asec:n6heis} and presented in the ancillary files.)
In one of the simplest cases, this recasting allows us to package a HEI into 2 terms as opposed to 13 terms when written in the S-basis.
Besides providing a powerful heuristic for obtaining valid inequalities, we expect this efficient rendition of HEIs to bring us one step closer to a structural and qualitative understanding of their physical meaning.

To make this expectation more precise, we attempt to interpret the non-negative information quantities associated to HEIs as correlation measures, and explore whether they are monotonically non-increasing under partial-tracing (equivalently, non-decreasing under inclusion), as is the case with the mutual information. The main result of our paper is a proof, in complete generality, that this is not possible.   In particular, we prove that \emph{none of the HEI information quantities besides the mutual information is monotonic under inclusion}.  In other words, we can always find some holographic configuration for which shrinking some subsystem causes the information quantity to increase.

Our proof utilizes an important property of HEIs called  \emph{superbalance}.  This was defined in \cite{Hubeny:2018ijt}, and subsequently in \cite{He:2020xuo} it was shown that all non-SA HEIs must be superbalanced.
This conveniently restricts the form that any HEI can take, so that when we consider the difference between the original information quantity and the corresponding quantity obtained by tracing out part of some subsystem, this difference can be shown to take either sign, even when restricted to the holographic context.  More concretely, showing that it is not holographically sign-definite proceeds by expressing the difference in terms of a sum over conditional multipartite informations and by judicious sequential choice of states demonstrating that we can collapse this sum onto just a single term which itself can take either sign.

The outline of this paper is as follows.
In \cref{sec:multiparty} we review the key properties of multipartite entanglement.  We start in \cref{ss:notation}  by setting up our notation and terminology.  Section \ref{ss:In} then defines the multipartite information $I_n$ and presents a number of lemmas summarizing the properties which we will utilize below, while the short \cref{ss:CIn} analogously presents conditional multipartite information.  All of these results can be derived straightforwardly from the basic definitions, but some of them simplify by using a certain structural property of the multipartite informations indicated in \cref{asec:Infactornotn}.  In \cref{sec:HIQ} we discuss the superbalance property of HEIs and its consequences, which paves the way to proving our main result in \cref{sec:correlations}, namely that superbalance precludes monotonicity under inclusion.  Having thus established that none of the higher HEIs act as correlation measures,  we summarize our explorations of recasting the HEIs into a more compact form in \cref{ss:HEIrecasting}, 
which also includes the tables of the non-uplift $\nN=5$ HEIs and selected representative subset of $\nN=6$ HEIs. 
This is accompanied by \cref{asec:heis} which explains how to obtain the compact forms starting from the standard representation.
We end with a brief summary and discussion of future directions in \cref{sec:discuss}.
Finally, to facilitate further explorations, in \cref{asec:n6heis} we briefly detail the process through which we have been able to generate more than $1800$ new HEIs for six parties, and summarize their forms in \cref{fig:HEIimfreq}.
Also, for comparison with v1, in \cref{tab:new_HEIs} we present 25 of the new HEIs (i.e.\ whose tripartite form has at least 5 terms), selected by typicality of their form.
This is supplemented with ancillary files, accessible on the arXiv page for this paper, where we list all the newly found inequalities.
A more complete data repository, including not just the facets but also the extreme rays known so far, is available in \cite{hecdata}.

\section{Multipartite Entanglement}
\label{sec:multiparty}

Hitherto, multipartite entanglement has remained a fascinating but poorly understood phenomenon at a quantitative level. The main obstacle resides already in the lack of meaningful measures of such correlations, or even a basic understanding of axiomatic properties these should reasonably satisfy. One of our goals in this paper is to explore the possibility of interpreting the information quantities associated to HEIs as multipartite correlation measures.
To do so, in this section we begin by developing key aspects of the multipartite information and conditional multipartite information,
which pave the way for a better reading of HEIs and their properties.  
To make the presentation notationally concise wherever possible, we will adopt several shorthand representations of the quantities of interest.

\subsection{Notation and setup}
\label{ss:notation}

The HEC is a geometrical construct in the entropy space $\mathbb{R}^\nD$ with axes labeled by subsystem entropies.  For $\N$-party configuration, there are $\nD=2^{\N}-1$ independent subsystems, consisting of all non-trivial collections of the individual parties.

We will label monochromatic subsystems (which we'll also refer to as \emph{singletons}) by either letters $\x{A,B,C,} \dots$, or when we wish to enumerate them, a letter with subscript $\x{A}_1,\x{A}_2,\x{A}_3$, etc..
As a notational shorthand, when working with $\N$-party systems, we will denote the set of all parties by $[\N]\equiv\{1,2,\dots,\N\}$.
A polychromatic system consisting of multiple monochromatic ones can be written out explicitly as e.g.\ $\x{AB}$.
When referring to a general nonempty subsystem without needing to specify its singleton composition, our notation will be
either a letter from the end of the alphabet such as $\x{X}, \x{Y}, \x{Z}$ already employed above, or (for the sake of adhering to previously-established notation)
$\pI$, $\pJ$, \dots, implicitly assumed to satisfy $\varnothing\ne\pI\subseteq[\N]$. We can then use this index as a subscript on entropy, $\S_\pI$, to denote the entropy of the joint polychromatic subsystem. Sometimes we will also wish to enumerate multiple polychromatic subsystems, in which case we will use a letter from the end of the alphabet with a subscript, e.g.\ $\x{X}_1,\x{X}_2,\x{X}_3$, etc..

At times it will also be useful to refer to the complementary subsystem which purifies the state for $\N$ parties at hand. Labelling the purifier of an $\N$-party system by $\N+1$, one may equivalently consider pure states on $\N+1$ parties. Subsystems which may involve the purifier will be decorated with an underline, $\pII$, $\pJJ, \ldots $, and throughout assumed to obey $\varnothing\ne\pII\subset[\N+1]$. An important property of the von Neumann entropy associated to this generalized notation is the fact that $\S_{\pII}=\S_{[\N+1]\smallsetminus\pII}$ for all $\pII$, namely, that by purity of the full state, the entropy of a subsystem and that of its complement are identical.

The permutation symmetries that the HEC enjoys will be relevant to the study of general information quantities as well. For $\N$ parties, the HEC is obviously symmetric under permutations of the party labels in $[\N]$, which are captured by the symmetric group $Sym_{\N}$. These are referred to as the permutation symmetry of the HEC. In addition, by the purity property mentioned above, the full symmetry group of the HEC gets enhanced to $Sym_{\N+1}$ for permutations of $[\N+1]$ under the identification $\S_{\pII}=\S_{[\N+1]\smallsetminus \pII}$ for all complementary subsystems. The set of permutations involving $\N+1$ are referred to as the purification symmetry of the HEC. Canonically, we will always use the purification symmetry to label entropies by the $2^{\N}-1$ subsets $\varnothing\ne \pI \subseteq[\N]$ which do not involve $\N+1$.

When referring to multipartite information with $n$ arguments, we write this abstractly as $I_n$; to indicate the explicit arguments, it is conventional to write them out explicitly as in 
$I_n(\x{X}_1:\x{X}_2:\ldots:\x{X}_n)$ where the $\x{X}_i$'s are mutually disjoint subsystems which can be monochromatic or polychromatic.  When considering a collection of just monochromatic arguments, it will be notationally convenient to compress them into a single (decorated) polychromatic subsystem, which we will denote by  an underdot like $\aset{\pI}$ or a dotted underline as in ${\x{\aset{AB}}}$.  
A further special case is when \emph{all} the $I_n$ arguments are monochromatic; since this is an important construct, we will use a dedicated notation and terminology for this case: we will call such an $I_n$ a \emph{singleton multipartite information} and emphasize this feature by a special font $\sI_n$.  Moreover if we want to specify the actual arguments, we can now use the combined polychromatic system as an index, in which case we drop the underdots and write simply $\sI_\pI$.  (Note that in the latter case we can leave the number of arguments $n = \abs{\pI}$ implicit.)   So for instance we can write interchangeably
\begin{equation}\label{eq:Innotations}
	 \sI_4(\x{A:B:C:D}) 
	 = \xI{ABCD}  
	 = \sI_4(\x{\aset{AB}:C:D}) = \sI_4(\x{A:\aset{BCD}})
\end{equation}	
Although redundant, allowing ourselves this notational flexibility will enable us to present the arguments and HEI tables more cleanly.

\subsection{Multipartite information}
\label{ss:In}

We now review the various salient properties of the multipartite information $I_n$.  We start with the definition and structural aspects of these information quantities, and then proceed to discuss the values they can evaluate to for given classes of configurations.

\paragraph{$I_n$ in terms of entropies:}
Using the notation established in \cref{ss:notation}, we can now write the multipartite information in terms of entanglement entropies quite compactly.  Let us first specialize to the case where all the arguments are monochromatic.  If the arguments comprise the polychromatic system $\pI$, we can write the singleton multipartite information as
\begin{equation}\label{eq:IfromS}
	\sI_\pI = \sum_{\pK \subseteq \pI} (-1)^{\abs{\pK}+1} \, \S_\pK \ ,
\end{equation}	
where the usual sign convention is fixed so as to keep the singleton coefficients positive.
For example, for the tripartite information this would give 
\begin{equation} 
	\xI{ABC} = \xS{A} + \xS{B} + \xS{C} - \xS{AB} - \xS{AC} - \xS{BC} + \xS{ABC}
\end{equation}	
and so forth.  
Once written out in the form \cref{eq:IfromS}, we can straightforwardly replace the monochromatic arguments by polychromatic ones, e.g.\ 
\begin{equation}
	I_3(\x{A:B:CD}) = \xS{A} + \xS{B} + \xS{CD} - \xS{AB} - \xS{ACD} - \xS{BCD} + \xS{ABCD}
\end{equation}	
(Note the change of font on the $I_3$ indicates that this is no longer a singleton multipartite information.)

\paragraph{$\sI_n$ basis:}
Let us consider the set of all singleton multipartite informations $\sI_\pI$ (i.e.\ those with only monochromatic arguments).  There are precisely $\nD$ of them, the same as the number of independent entropies, and hence the dimensionality of the entropy space.
Indeed, as argued in \cite{Hubeny:2018ijt,He:2019ttu}, the set of all $\sI_\pI$'s form a basis for the entropy space, referred to as the ``I-basis'', to distinguish it from the ``S-basis'' representation in terms of the entropies.  This means that we can express any information quantity in terms of these multipartite informations.  For example, the entropies themselves are obtained very simply by inverting \cref{eq:IfromS}, which is in fact captured by an involutory matrix:
\begin{equation}\label{eq:SfromI}
	\S_\pI = \sum_{\pK \subseteq \pI} (-1)^{\abs{\pK}+1} \, \sI_\pK	 \ .
\end{equation}	

\paragraph{$I_n$ in terms of other $I_m$'s:}
It will be useful to express the multipartite information with polychromatic arguments in the I-basis as well, and conversely, to express $\sI_n$ in terms of $I_m$'s with smaller $m$ but polychromatic arguments.  To obtain such relations, we will make use of the underlying symmetries of multipartite information.
From \cref{eq:IfromS} we can easily see that any $\sI_\pI$ is symmetric under permuting its arguments.  In fact it has a much larger symmetry since all the coefficients are $\pm 1$ depending on the parity of the size of the term.  Indeed, this has an iterative structure which is manifest in a ``notationally-factorized'' form presented in \cref{asec:Infactornotn}.
This notation enables us to obtain the following relations straightforwardly.

First of all, we can build up higher multipartite informations, say $I_{n+1}$, from a combination of lower ones, say $I_n$'s, but at the cost of invoking polychromatic arguments.  In particular, using the shorthand notation indicated above in \cref{eq:Innotations} with $\aset{\pI}$ comprising $n-1$ singletons distinct from $\x{A}$ or $\x{B}$, we have, for any $n$,
\begin{equation}\label{eq:Inadditivitysplit}
	\sI_{n+1}(\x{A}:\x{B}:\aset{\pI}) 
	= \sI_n(\x{A}:\aset{\pI})+  \sI_n(\x{B}:\aset{\pI}) - I_n(\x{AB}:\aset{\pI})
\end{equation}	
which by permutation symmetry on the arguments of $I_n$ applies for any positions of $\x{A}$ and $\x{B}$.  We can furthermore ``fine-grain'' any singleton(s) into polychromatic subsystem(s), to obtain any $I_{n+1}$ in terms of $I_n$'s.
We can then iterate \cref{eq:Inadditivitysplit} further, to express any $I_n$ in terms of a collection of $I_m$'s for any smaller $m$, so in particular there is a multitude of ways we can express any $I_m$ purely in terms of just mutual informations $I_2$ between various polychromatic subsystems.
The special case of $n=2$ was already utilized in \cref{eq:MMI}.  Conversely, we can iterate \cref{eq:Inadditivitysplit} to express multipartite information with polychromatic arguments in the I-basis.
We will make explicit use of these relations to obtain a compact recasting of the HEIs in \cref{asec:heis}.

\paragraph{$I_n$ values:}

Having defined the multipartite information $I_n$ in the preceding subsection, 
we now consider the actual values it can take for a given configuration.  
Recall that by SA (\cref{eq:SA}), the mutual information $I_2$ for any two subsystems is always non-negative and vanishes when the subsystems are uncorrelated, and by MMI (\cref{eq:MMI}) the tripartite information $I_3$ between any three subsystems is non-positive and likewise vanishes when the subsystems are uncorrelated.
In this subsection we will show that while the sign-definiteness does not extend to any higher $n$ (for $n=4,5$ this was already observed in \cite{Hayden:2011ag} and noted for general $n$ in \cite{Hubeny:2018ijt}), the feature of vanishing on decorrelated states does generalize to all $I_n$'s.  

We start with the latter part which is more immediate.  Of course if all the $n$ subsystems invoked by $I_n$ were pairwise decorrelated from each other, then we could rewrite $I_n$ in terms of just $I_2$'s and apply the above SA saturation result to each term.  We can however drastically relax the requirement to obtain a much stronger statement:  for the $n$-partite information $I_n$ to vanish, it suffices for any subsystem comprising a subcollection of its arguments to be decorrelated from the remainder.

\begin{nlemma}
\label{lem:Ivashprod}
\emph{$I_n(\x{X}_1:\ldots:\x{X}_n)$} vanishes on any product state of the form \emph{
$$\rho_{_{\x{X}_1\dots \x{X}_n}} = \pi_{_{\mathcal{J}}} \otimes \sigma_{_{\mathcal{K}}},$$} with $\mathcal{J}$ and $\mathcal{K}$ composed of complementary subsets of the \emph{$\x{X}_i$} subsystems.
\end{nlemma}
\begin{proof}
For this proof, polychromatic subsystems $\x{X}_i$ can be treated as irreducible monochromatic subsystems $\x{A}_i$, since the sub-parties of each $\x{X}_i$ will play no role.
Therefore, consider $I_{[n]}$ evaluated on some state $\rho_{[n]}$. A product state $\rho_{[n]}$ admits a bipartition, i.e., $[n] = \mathcal{J} \cup \mathcal{K}$ with $ \mathcal{J} \cap \mathcal{K} = \varnothing$, such that 
\begin{equation}
    \rho_{[n]} = \pi_{_{\mathcal{J}}} \otimes \sigma_{_{\mathcal{K}}} \, .
\end{equation}
Entropies evaluated on $\rho$ then satisfy
\begin{equation}
\label{eq:SA-saturation}
    \S(\pJ \cup \pK) = \S(\pJ) + \S(\pK), \qquad \forall \pJ \subseteq \mathcal{J}, \pK \subseteq \mathcal{K}.
\end{equation}
When substituted into \cref{eq:IfromS}, it is easy to verify that $I_{[n]}$ reduces to
\begin{equation}
    I_{[n]}\big|_{\rho} = \alpha(\mathcal{K}) \, I_{\mathcal{J}}\big|_\pi + \alpha(\mathcal{J}) \, I_{\mathcal{K}}\big|_\sigma
\end{equation}
with coefficient $\alpha(\cdot)$ given by
\begin{equation}
    \alpha(\pI) =  \sum^{|\pI|}_{k=0} (-1)^k \binom{|\pI|}{k} = 0 \, . 
\end{equation}
This shows that the multipartite information vanishes on any system whose density matrix can be written as a product between the density matrices of two systems.
\end{proof}

Let us now consider the case when the composite argument in $I_n$ comprises a density matrix which does not factorize (or factorizes in a way which does not `align' with the $n$-partition considered above).   Does the sign of $I_n$ get fixed for any $n>3$?  To see that the answer is \emph{no}, it suffices to produce two configurations, one for which $I_n$ is positive and another for which it is negative.
While we could certainly do this with explicit holographic geometries, it is even easier to produce such examples using the toolkit of holographic graph models.\footnote{ \, As first done in \cite{Bao:2015bfa}, the problem of computing entropies using the Ryu-Takayanagi (RT) \cite{Ryu:2006bv} formula for arbitrary subsystems in a given holographic configuration can be concisely recast into a graph-theoretic language. Deploying the RT surfaces for all relevant subsystems altogether, they partition the bulk into a set of codimension-$0$ regions. One assigns a vertex to every such region, and connects them by an edge whenever they correspond to adjacent bulk regions. Each edge is given a weight corresponding to the area of the associated piece of RT surface they go through. Those vertices on regions reaching the conformal boundary are understood as capturing the boundary subsystems. This way, one ends up with a desiccation of the bulk geometry, where all results obtained from the RT prescription of minimizing areas can be efficiently reproduced by the prescription of minimizing cut weights on the resulting holographic graph model.}

\begin{figure}[h!]  
  \begin{center} 
  \begin{tikzpicture}
	\tikzmath{
		\v =0.1;  
		\e=1.8;  
	}
	\coordinate (bv) at (0,0);  
	\coordinate (Av) at (90:\e);  
	\coordinate (Bv) at (30:\e);  
	\coordinate (Cv) at (-30:\e);  
	\coordinate (Dv) at (210:\e);  
	\coordinate (Ev) at (150:\e);  
	\draw[edgestyle] (Av) -- node[blue,left,pos=0.4]{$1$} (bv);
	\draw[edgestyle] (Bv) -- node[blue,above,pos=0.4]{$1$} (bv);
	\draw[edgestyle] (Cv) -- node[blue,above,pos=0.4]{$1$} (bv);
	\draw[edgestyle] (Dv) -- node[blue,below=-0.2,pos=0.55]{$1$} (bv);
	\draw[violet,thick] (Ev) -- node[violet,below,pos=0.4]{$w_0$} (bv);
	\filldraw [color=gray] (bv) circle (\v) ;
	\filldraw [color=black] (Av) circle (\v) node[black,above right=0pt]{$\x{A}_1$} ;
	\filldraw [color=black] (Bv) circle (\v) node[black,right=2pt]{$\x{A}_2$} ;
	\filldraw [color=black] (Cv) circle (\v) node[black,right=2pt]{$\x{A}_3$} ;
	\filldraw [color=black] (Dv) circle (\v) node[black,left=2pt]{$\x{A}_n$} ;
	\filldraw [color=black] (Ev) circle (\v) node[black,left=2pt]{$\x{A}_{n+1}$} ;
    \draw[loosely dotted,very thick] (-65:\e) arc (-65:-115:\e]);
\end{tikzpicture}  
    \caption{Star graph with $n$ edge weights fixed to $1$ and the remaining edge weight to some $w_0>0$.
    Depending on the value of $w_0$, $\sI_{[n]}$ evaluated on this graph can take either sign.}
      \label{fig:star}
  \end{center}
\end{figure}
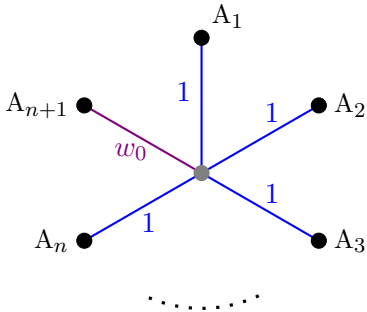  

\begin{nlemma}
\label{lem:Inosign}
    $I_n$ with $n\geq4$ is not holographically sign-definite.
\end{nlemma}
\begin{proof}
    To prove this claim, it suffices to show that it holds on graphs for $n$ parties. This is because for any larger number of parties, one could just assign the extra parties to the purifier, thereby effectively reducing it back to an $n$-party problem. Consider a star graph with $n+1$ edges, involving all parties as well as the purifier. Assign weight $w_0$ to the purifier edge, and unit weight to the rest as indicated in \cref{fig:star}.\footnote{ \, This family of star graphs was also fruitful in proving structural properties of the HEC in \cite{Avis:2021xnz,Czech:2021rxe,Fadel:2021urx}.}
    The entropy of any nonempty subsystem $\pK \subseteq[n]$ on such a star graph is:
    \begin{equation}
        \S_\pK = \min \{ \abs{\pK} \,,\, w_0 + n - \abs{\pK} \} \, .
    \end{equation}
    Since this only depends on the cardinality of $\pK$, our quantity of interest $\mathcal{I}_n \equiv \sI_{[n]}$ becomes
    \begin{equation}
    \begin{aligned}
        \mathcal{I}_n &= -\sum_{k=1}^{n} \, \binom{n}{k} (-1)^{k} \min \{ k \,,\, w_0 + n - k \} \\
        &= -\sum_{k=1}^{\lfloor \frac{w_0+n}{2} \rfloor} \binom{n}{k} (-1)^{k} \, k - \sum_{k=\lceil\frac{w_0+n}{2} \rceil}^{n} \binom{n}{k} (-1)^{k} \, (w_0+n-k)
    \end{aligned}
    \end{equation}
    Letting $w_0=n-1$, this evaluates to
    \begin{equation}
    \label{eq:starsw1}
        \mathcal{I}_n = (-1)^n \, n - (-1)^n \, (n-1) = (-1)^n.
    \end{equation}
    On the other hand, letting $w_0=n-3$ leads to 
    \begin{equation}
        \mathcal{I}_n = -(-1)^n \, n (n-2) + (-1)^n \, (n (n-3)+3) = - (-1)^{n} \, (n-3).
    \end{equation}
    Hence $\mathcal{I}_n$ takes opposite signs in each of these star graphs, proving the claim.\footnote{ \, Note that the RHS vanishes for $n=3$, consistent with MMI.}
\end{proof}

\paragraph{Multipartite information for the entire system:}
The above two lemmas considered the typical case where the arguments of $I_n$ comprise a polychromatic subsystem $\pI \subseteq [\nN]$.  We now complete the characterization by considering what happens when the arguments comprise the entire purified system $[\nN+1]$. 
For the mutual information ($n=2$ case), the result is particularly simple: denoting the first argument by $\pI$, the second argument is then its complement $\cplmt{\pI} = [\nN+1]\smallsetminus \pI$, so applying the definition in \cref{eq:SA} and using $\S_\pI = \S_{\cplmt{\pI}}$ and $\S_{[\nN+1]} = 0$, we have
\begin{equation}\label{eq:I2full}
I_2(\pI:\cplmt{\pI})= 2 \, \S_\pI
\end{equation}
where we could think of the RHS as $2 \, I_1(\pI)$.  A similar result holds for every even $n$. Schematically, $I_n = 2 \, I_{n-1}$, where any one of the arguments on the LHS gets dropped for the RHS expression (and therefore the latter is automatically purification-symmetric). The result for odd $n$ is even simpler since there we have complete cancellation of terms, giving $I_n = 0$.

We summarize and prove these claims in the following lemma:
\begin{nlemma}
\label{lem:Ifullsystem}
    Consider the $n$-partite information $I_n$ evaluated on $n$ disjoint subsystems \emph{$\x{X}_i$} comprising a full $\N$-party system including the purifier, \emph{$\bigcup_i \x{X}_i = [\N+1]$}.\footnote{ \, We remark that here each $\x{X}_i$ subsystem may involve more than one of the parties in $[\N+1]$, so that the statement applies to a general $n\leq \N+1$ and not just the $n=\N+1$ case where each $\x{X}_i$ is a single party.}  Then
    \emph{
	\begin{equation}\label{eq:Inevenoddfull}
	I_n (\x{X}_1:\ldots:\x{X}_n) = 
		\begin{cases}
		2 \, I_{n-1}(\x{X}_1:\ldots:\x{X}_{n-1}) \ , \qquad & n \ \text{even} \\
		0 \ , \qquad  & n \ \text{odd} \\
		\end{cases}
	\end{equation}
    }
\end{nlemma}
\begin{proof}
In both cases, one can show the result by direct evaluation.  Since   $\bigcup_i \x{X}_i = [\N+1]$, one of the $\x{X}_i$'s contains the purifier; w.l.o.g., let us assume that it is in $\x{X}_n$.  Then we can replace  $\x{X}_n = [\N+1] \setminus \bigcup_{i=1}^{n-1} \x{X}_i$, and similarly any subsystem containing $\x{X}_n$ by the corresponding complementary subsystem.  When we re-express $I_n (\x{X}_1:\ldots:\x{X}_n)$ in terms of entropies using \cref{eq:IfromS}, a given combination of $\x{X}_i$'s with $i\in [n-1]$ appears twice in the resulting expression, with coefficient $\pm1$ depending on the number of terms combined for the original expression and for its complement.  For odd $n$, these coefficients are opposite, and therefore all terms cancel (except for the entropy of the full system, which however vanishes by purity, $\S(\x{X}_1 \ldots \x{X}_n) =0$).  For even $n$ both coefficients are the same, so they add constructively, giving the stated result.
\end{proof}

\paragraph{Geometrical significance of the $I_n$'s:}
So far we have discussed various useful properties of the multipartite information quantities, which will come into play both for proving our no-go theorem in \cref{sec:correlations} as well as for presenting the HEIs in a more compact form in \cref{ss:HEIrecasting}.  As a further motivation for considering the $I_n$'s, we note that they also provide a useful diagnostic of the basic geometrical features of the bulk configuration, in particular the connectivity of the entanglement wedges of the various subsystems.  
Let us characterize a bulk entanglement wedge of a composite subsystem $\x{AB}$ as \emph{connected} if there exists a continuous curve $\gamma$ through the bulk joining the boundary regions $\x{A}$ and $\x{B}$ which stays entirely within their entanglement wedge.\footnote{ \, 
	By entanglement wedge we mean the full codimension-0 bulk region as originally defined in \cite{Headrick:2014cta}, namely the bulk domain of dependence of the corresponding homology region (bounded by the said boundary region and its HRT surface \cite{Hubeny:2007xt}).  However, since the regions $\x{A}$ and $\x{B}$ must be spacelike-separated in order for $\xI{AB}$ to be well-defined, it is simple to show that $\gamma$ can be chosen to be spacelike, and the following connectivity statement can be reformulated to pertain to any Cauchy slice containing $\x{A}$ and $\x{B}$ on its boundary.
}
Then generically (i.e.\ when each subsystem has a unique HRT surface), the entanglement wedge of a composite subsystem $\x{AB}$ is connected between $\x{A}$ and $\x{B}$ if and only if the two subsystems are correlated ($\xI{AB} > 0$).  
This follows straightforwardly from the observation that $\xI{AB}=0$ iff the composite HRT surface for $\x{AB}$ consists of the individual HRT surfaces for $\x{A}$ and $\x{B}$, thereby disconnecting the joint entanglement wedge.  

This statement can be suitably generalized to use $I_n$ to characterize simultaneous connectivity of $n$-party entanglement wedges.  We define an $n$-party entanglement wedge of $\x{X}_1\ldots \x{X}_n$ as connected if there exists a bulk point $p$ along with $n$ continuous curves $\gamma_{pi}$ from $p$ to the respective boundary regions $\x{X}_i$, all contained entirely within the entanglement wedge.\footnote{ \, 
	Note that according to this definition, a connected entanglement wedge may still consist of multiple disjoint pieces (which could happen if at least one of the regions $\x{X}_i$ is comprised of multiple components), provided at least one of those pieces contains the full set $\{\gamma_{pi}\}$.
	If, on the other hand, each $\x{X}_i$ admits only a single entangling surface, then the entanglement wedge consists of a single piece and connectivity is tantamount to the existence of a set of curves $\gamma_{ij}$ joining all pairs of regions $\{\x{X}_i,\x{X}_j\}$ to each other within it.
}
Then as a consequence of \cref{lem:Ivashprod}, we have the following connectivity statement:

\begin{nprop}
\label{prop:InEWconnectivity}
    The multipartite information of $n$ disjoint subsystems \emph{$\x{X}_i$} in a generic\footnote{ \, 
    	By ``generic'' we mean that we have not fine-tuned to a `phase transition' when multiple extremal surfaces exchange dominance; in other words, each subsystem has a unique HRT surface, and these have independent areas with no accidental cancellations.
    } configuration is non-vanishing,
    \emph{
	\begin{equation}\label{eq:Inevenoddfull2}
	I_n (\x{X}_1:\ldots:\x{X}_n) \ne 0 \ ,
	\end{equation}
    }
    if and only if the joint entanglement wedge of \emph{$\x{X}_1\ldots \x{X}_n$} is connected.
\end{nprop}
\begin{proof}
	To warm-up, let us first consider the simple case where each region $\x{X}_i$ has a single entangling surface.  
	Suppose the entanglement wedge of $\x{X}_1\ldots \x{X}_n$  is disconnected.  Then there exists a bipartition of the system $\bigcup_i \x{X}_i = \pI \cup \pK$ such that the entanglement wedge of $\pI \pK$ is disconnected, and by the previous argument this means that $\mi{\pI}{\pK}=0$.  Hence the corresponding density matrix $\rho_{\pI \pK}$ factorizes, and by \cref{lem:Ivashprod}, $	I_n (\x{X}_1:\ldots:\x{X}_n) = 0$.

	However, there can be situations where the joint entanglement wedge is disconnected, but where all bipartitions have connected entanglement wedges.  A simple example of this is if the configuration has the entanglement structure of a collection of Bell pairs, such as indicated in fig.~9 of \cite{Hubeny:2018trv}.  In this case, we can first fine-grain the full system so as to have the disjoint pieces of the entanglement wedge correspond to disjoint subsystems, and apply the above reasoning to derive the vanishing of the corresponding fine-grained multipartite information.  Then by the same calculation as in \cref{lem:Ivashprod}, the coarse-grained $	I_n (\x{X}_1:\ldots:\x{X}_n) $ combines additively from the fine-grained one, leading to the desired result $I_n (\x{X}_1:\ldots:\x{X}_n) =0$.  This proves (the contrapositive of) the ``only if'' ($\Rightarrow$) part of the statement.

	Conversely, suppose the entanglement wedge of $\x{X}_1\ldots \x{X}_n$  is connected.  We can again bootstrap our argument.  In the  simple case where each region $\x{X}_i$ is connected and has a single entangling surface, connectivity of the entanglement wedge implies that there is a connected HRT surface for $\x{X}_1\ldots \x{X}_n$, whose area is therefore generically independent of the other HRT surfaces for the various subsystems, and therefore will not cancel in the $I_n$ expression.  In the more general case of the  $\x{X}_i$'s consisting of multiple components, even if the HRT surface for $\x{X}_1\ldots \x{X}_n$ consists of multiple components, connectivity of the entanglement wedge guarantees that there is still at least one extremal surface which is anchored on at least one component of each $\x{X}_i$, and therefore again cannot be ``canceled" by any of the subsystem HRT surfaces.  This proves the ``if'' ($\Leftarrow$) part of the statement.
\end{proof}

In using the above arguments, it is important to realize that individual mutual information quantities are independent of the multipartite information; for example, there can be configurations for which
\begin{equation}\label{eq:In0I2}
	I_2 (\x{X}_i:\x{X}_j) > 0 \quad \forall \ i,j
	\qquad \text{but} \qquad
	I_n (\x{X}_1:\ldots:\x{X}_n) =0 \ ,
\end{equation}
as well as configurations for which
\begin{equation}\label{eq:InI20}
	I_2 (\x{X}_i:\x{X}_j) =0 \quad \forall \ i,j
	\qquad \text{but} \qquad
	\abs{I_n (\x{X}_1:\ldots:\x{X}_n)} > 0 \ .
\end{equation}
(Examples of both of these cases were presented e.g. in\ \cite{Hubeny:2018trv}, but an analogous graph-based example is e.g.\ a collection of all $ij$ Bell pairs for the former and a uniform star graph of \cref{fig:star} with $w_0=1$ for the latter.)
It is the bipartitioning of the full system which is related to the multipartite information.  The geometrical significance of the $I_n$ is the simultaneous connectivity within the single entanglement wedge; in other words, $\abs{I_n}>0$ diagnoses the presence of a bulk region which lies inside the joint entanglement wedge but outside all possible subsystem entanglement wedges.

\subsection{Conditional multipartite information}
\label{ss:CIn}

Now that we have explored the properties of multipartite information $I_n$, let us consider a related but slightly more intricate construct, namely that of \emph{conditional} multipartite information.
The motivation for this comes from the  observation that conditional quantities have a natural quantum information theoretic interpretation.  The conditional entropy, where we condition $\Sx{X}$ on $\x{Y}$, is defined as $S(\x{X}|\x{Y}) \equiv \Sx{XY}-\Sx{Y}$, and measures the amount of quantum communication needed to obtain complete quantum information contained in the joint state $\rhox{XY}$ if we already know the marginal $\rhox{Y}$.
Classically, this quantity is non-negative, but quantum mechanically it can take either sign.\footnote{ \, When this quantity is negative, its absolute value can be interpreted as the additional amount of quantum information which $\x{X}$  can transmit to $\x{Y}$ ``for free'' (i.e.\ using only classical communication).
}
On the other hand, for any quantum state the conditional mutual information is non-negative by SSA, cf.\ \cref{eq:SSA}.  While we will show momentarily that the higher $n$ conditional multipartite information is no longer sign definite, it nevertheless seems to provide a nicer packaging for the HEIs, as discussed in \cref{ss:HEIrecasting}.

In general, a conditional information quantity (e.g.\ conditioned on subsystem $\x{A}$) can be obtained by writing it out as a linear combination of entropies $\S_\pI$ with $\pI \nsupseteq \x{A}$ and replacing each entropy $\S(\pI)$ by the conditional entropy $S(\pI | \x{A})$.
Using the notation from \cref{eq:Innotations}, for $\aset{\pI}$ comprising $n-2$ singletons distinct from $\x{A}$, $\x{B}$, or $\x{C}$, we have two useful identities for conditional multipartite information $I_n(\x{B}:\aset{\pI} \vert \x{A}) $, viewed as the multipartite information $I_n(\x{B}:\aset{\pI})$ conditioned on $\x{A}$:
\begin{equation} \label{eq:CondMultInf}
\begin{aligned}
    I_n(\x{B}:\aset{\pI} \vert \x{A}) 
    &=  I_n(\x{AB}:\aset{\pI}) - I_n(\x{A}:\aset{\pI}) \\
    &=  I_n(\x{B}:\aset{\pI}) - I_{n+1}(\x{A}:\x{B}:\aset{\pI})
\end{aligned}
\end{equation}
where we have used the definition of conditioning in the first line and \cref{eq:Inadditivitysplit} in the second line.
We now use these relations to show the following:

\begin{nlemma}
\label{lem:condInosign}
Conditional $I_n$ with $n\geq3$ is not holographically sign-definite.
\end{nlemma}
\begin{proof}
We first consider the case $n\geq 4$,
where by a judicious choice of the state, we can reduce the argument to that of \cref{lem:Inosign}.
We rewrite conditional $I_n$ as
\begin{equation} \label{eq:CondMultInfA}
    I_n (\x{A}_1:\ldots:\x{A}_n | \x{A}_{n+1}) 
    = \sI_{[n]} - \sI_{[n+1]}
\end{equation}
and consider a product state
\begin{equation} 
    \rho_{\x{A}_1\dots \x{A}_n \x{A}_{n+1}} = \pi_{\x{A}_1 \dots \x{A}_n} \otimes \sigma_{\x{A}_{n+1}}.
\end{equation}
In such a state the last term of \cref{eq:CondMultInfA} vanishes by \cref{lem:Ivashprod}, so we have
\begin{equation} \label{eq:CondMultInfAfact}
    I_n (\x{A}_1:\ldots:\x{A}_n | \x{A}_{n+1})\big|_\rho 
    =\sI_{[n]}\big|_\pi \ .
\end{equation}
Since we placed no restriction on $\pi_{\x{A}_1 \dots \x{A}_n}$,
\cref{lem:Inosign} then implies that the RHS of \cref{eq:CondMultInfAfact} can take either sign, so conditional $I_n$ with $n\geq4$ is not holographically sign-definite. 

Now we consider the $n=3$ case. Consider w.l.o.g.\ a pure state on $\x{A}_1\dots \x{A}_5$. Purification symmetry of $I_3$ implies (cf.\ \cref{eq:CI3swap})\footnote{ \, 
Written out explicitly, $I_3(\x{A}_1:\x{A}_2:\x{A}_3|\x{A}_4) = I_3(\x{A}_1:\x{A}_2:\x{A}_3\x{A}_4) - I_3(\x{A}_1:\x{A}_2:\x{A}_4) = I_3(\x{A}_1:\x{A}_2:\x{A}_5) - I_3(\x{A}_1:\x{A}_2:\x{A}_3\x{A}_5) = -I_3(\x{A}_1:\x{A}_2:\x{A}_3|\x{A}_5)$.}
\begin{equation}
    I_3(\x{A}_1:\x{A}_2:\x{A}_3|\x{A}_4) = -I_3(\x{A}_1:\x{A}_2:\x{A}_3|\x{A}_5) \, .
\end{equation}
Therefore conditional $I_3$ being sign-definite would imply conditional $I_3$ vanishes on any state for any choice of regions, which is obviously not true.\footnote{ \, Since this argument only uses purification symmetry of $I_3$, it also works for all the odd $n$ cases as a corollary of \cref{lem:Ifullsystem}.}
\end{proof}

\section{HEIs and Superbalance}
\label{sec:HIQ}

In this section we review a crucial property, dubbed superbalance, pertaining to all holographic information quantities describing the non-SA facets \cite{He:2020xuo}.

A general information quantity $Q$ can be written as a linear combination of subsystem entropies $\S_\pI$, as
\begin{equation}
\label{eq:IQ}
    Q = \sum_{\pI \subseteq [\N]} c_\pI \, \S_\pI
\end{equation}
with some rational coefficients $c_\pI$ which can take either sign.

Although the $\nD=2^{\N}-1$ individual entropies $\S_\pI$ in \cref{eq:IQ} are independent of each other, it is interesting to consider a specific subset of terms which involve a given subsystem $\pJ$ by considering just the terms with $\pI \supseteq \pJ$.
We say that $Q$ is \emph{balanced in $\pJ$} if the appearance of this subsystem ``cancels out'' in $Q$ in the sense that
\begin{equation}
\label{eq:Jbalance}
    \sum_{\pI\supseteq\pJ} c_\pI = 0 \, .
\end{equation}
Typically, we think of balance just in the singletons $\pJ=\{j\}$ for $j\in[\N]$, and say that $Q$ is \textit{balanced} if it is balanced for every $j\in[\N]$, i.e., if
\begin{equation}
\label{eq:balance}
    \sum_{\pI\ni j} c_\pI = 0 \qquad \forall~j\in[\N] \, .
\end{equation}
This property is important when looking at entropy inequalities, in which case it corresponds to having each party showing up in the  entropy terms on one side as many times as on the other when the inequality is canonically written involving only positive signs. 
The operational significance of balance is that for any configuration of boundary subsystems corresponding to non-adjoining regions, a balanced quantity is guaranteed to be UV-finite, since all the UV divergences cancel out.\footnote{ \,
	For adjoining regions the mutual information is no longer finite, though this UV divergence is well understood \cite{Casini:2015woa}, and arises from the UV correlations across the joint entangling surface.
}

While well-defined on specific information quantities, the property of balance can be easily seen to be preserved under $Sym_{\N}$ but not under the purification symmetry in general.\footnote{ \, 
    For example, the mutual information $\xS{A}+\xS{B}-\xS{AB}$ is clearly balanced. However, e.g.\ for $\N=2$ we may use $Sym_{3}$ to exchange $\x{B}$ with the purifier $O$. Then, using purity to canonically write $\xS{O}=\xS{AB}$ and $\xS{AO}=\xS{B}$, one ends up with $\xS{A}+\xS{AB}-\xS{B}$ which is clearly not balanced.
} This is undesirable since, given the $S_{\N+1}$ symmetry of the HEC, we would like to probe its structure using objects whose properties are invariants of $Sym_{\N+1}$ as well. A property of $Q$ is invariant under $Sym_{\N+1}$ if every quantity in the symmetry orbit of $Q$ is, in which case it becomes a property of the symmetry orbit itself. This originally motivated \cite{Hubeny:2018ijt} to upgrade balance to a property of symmetry orbits by demanding that it be preserved under the purification symmetry, which suggests the following definition:
\begin{ndefi}[Superbalance]
    An information quantity is said to be superbalanced if every element in its $S_{\N+1}$ symmetry orbit is balanced.
\end{ndefi}
This property can be expressed algebraically in a way that is very analogous to \cref{eq:balance} for balance. Namely, as shown in Lemma $4$ of \cite{Avis:2021xnz}, $Q$ is superbalanced if and only if
\begin{equation}
\label{eq:superbalance}
    \sum_{\pJ\supseteq\pI} {c_{\pJ}} = 0 \qquad \forall~\varnothing\ne {\pI}\subseteq[\N]~~\text{with}~~\abs{\pI}\le 2 \, .
\end{equation}
In other words, $Q$ is superbalanced if it is balanced not only in all singletons but also in all doubletons.  

The I-basis is remarkably successful in making the properties of balance and superbalance manifest. To see this, consider re-expressing $Q$ in this basis as
\begin{equation}
\label{eq:relbases}
    Q = \sum_{\pI \subseteq [\N]} c_\pI\, \S_\pI = \sum_{\pJ \subseteq [\N]} q_\pJ \, \sI_{\pJ} \, ,
\end{equation}
Using \cref{eq:SfromI}, one can easily read off that the I-basis coefficients $q_\pJ$ can be written in terms of the S-basis coefficients $c_\pI$ through\footnote{ \, Using standard matrix notation so that \cref{eq:relbases} and \cref{eq:SfromI} respectively read $c^T \S=q^T \sI$ and $\S=M \, \sI$ for $M$ the transformation matrix between bases, one immediately finds $q=M^Tc$.}
\begin{equation}
\label{eq:relcoefbases}
    q_\pJ = (-1)^{\abs{\pJ}-1}\sum_{\pI\supseteq\pJ} c_\pI \, .
\end{equation}
Up to signs, this is of course nothing but the left-hand side of \cref{eq:Jbalance}, whose vanishing appears in the definitions of both balance and superbalance. Hence we obtain the following result (cf.\ Definition 10 of \cite{Hubeny:2018ijt}) 
\begin{nprop}
\label{prop:balsuperbal}
    An information quantity is balanced if and only if its representation in the I-basis does not involve any $\sI_1$ terms, and it is superbalanced if and only if it additionally does not involve any $\sI_2$ terms.
\end{nprop}
\begin{proof}
    Writing $Q$ as in \cref{eq:relbases} and using \cref{eq:relcoefbases}, the result follows from the defining properties of balance in \cref{eq:balance} and superbalance in \cref{eq:superbalance}.
\end{proof}
In \cite{Hubeny:2018ijt}, it was observed that any superbalanced information quantity is finite on all smooth configurations, even those involving adjoining regions where a merely balanced quantity would diverge.  To show this, it suffices to demonstrate finiteness for any $I_3$, since a superbalanced information quantity can be expressed as a finite sum of $I_3$'s.\footnote{ \, This follows from \cref{prop:balsuperbal}, and the fact that all I-basis elements $\sI_n$ for $n\geq3$ can be spanned by $I_3$'s with polychromatic arguments, as can be seen by applying \cref{eq:Inadditivitysplit} recursively.}  It is easy to check that the area divergences from common parts of entangling surface all cancel  in the combination of entropies given by $I_3$.  Indeed, since the subleading UV divergences (arising from corner terms etc.\ in higher-dimensional configurations) are local, they all likewise cancel (this was already argued for all $\sI_n$'s with $n\ge 3$ in \cite{Hayden:2011ag}), which establishes finiteness in even greater generality.

A powerful property of HEIs that is also made manifest by the I-basis is the requirement that the coefficients of the $I_n$ terms alternate in sign as follows \cite{Hubeny:2018ijt}: 
\begin{nprop}
\label{prop:relcoefbases}
    An information quantity defines a valid HEI only if, written as in \eqref{eq:relbases}, the I-basis coefficients satisfy
    $$(-1)^{\abs{\pK}} \, q_\pK \ge 0 \ , \qquad \forall \ \pK \subseteq[\nN]\, .$$
\end{nprop}
\begin{proof}
    Let $Q$ be an $\nN$-party information quantity. Recall the star graphs of \cref{fig:star}
    introduced in the proof of \cref{lem:Inosign}. For any subsystem $\pK\subset[\N]$, let $n=\abs{\pK}$ and consider a star graph with $n+1$ legs, $n$ of which are assigned unit weight and connected to vertices labelled by the parties in $\pK$, plus one of weight $w_0=n-1$ for the purifier $\nN+1$. Adding isolated vertices for any remaining parties in $[\nN]\smallsetminus\pK$, this constitutes a graph model for $\nN$ parties. Consider evaluating $Q$ on such a model. We can do so by evaluating every $\sI_{\pI}$ on it and using \cref{eq:relbases}. Since the graph model factorizes between $\pK$ and $[\nN]\smallsetminus\pK$, it follows by \cref{lem:Ivashprod} that $\sI_{\pI}$ vanishes for any $\pI\supset\pK$. Furthermore, for any $\pI\subset\pK$ one has $\S_{\pJ}=\abs{\pJ}$ for all $\pJ\subseteq\pI$ (and zero otherwise), which also implies the vanishing of $\sI_{\pI}$.\footnote{ \, Explicitly, $\sI_{\pI} = \sum_{\pJ\subseteq\pI} (-1)^{\abs{\pJ}} \abs{\pJ} = \sum_{j=1}^{\abs{\pI}} \binom{\abs{\pI}}{j} (-1)^{j} j = 0 $.} Lastly, for $\pI=\pK$, we obtained in \cref{eq:starsw1} that $\sI_\pK = (-1)^{\abs{\pK}}$. In other words, we have constructed a graph model for each $\pK\subseteq[\N]$ that gives $\sI_{\pI}=0$ for all $\pI\neq\pK$ and such that $Q$ evaluates to 
    $$Q=(-1)^{\abs{\pK}} \, q_\pK \, .$$ 
    Since a valid HEI for $Q$ would read $Q\geq0$ for any graph model, the claim follows.
\end{proof}

What makes the superbalance property so crucial in understanding meaningful holographic information quantities and the structure of the HEC itself is the following result proven in \cite{He:2020xuo}.
\begin{nthm}
\label{thm:superbalhec}
   Except for subadditivity, every facet of the HEC is superbalanced.
\end{nthm}
It follows from this and \cref{prop:balsuperbal} that the set of normal vectors defining the facets of the HEC other than the ones for SA live in a proper subspace of entropy space. In particular, since there are $\N$ entropy terms $\sI_1$ and $\binom{\N}{2}$ mutual information terms $\sI_2$ which vanish for all of them, they all belong to a subspace of codimension $\frac{1}{2}\N(\N+1)$. Geometrically, what this means is that the HEC naturally decomposes into two complementary subspaces, an SA subspace (spanned by the normals of SA facets) and a superbalanced one (spanned by the normals of all other facets). From the viewpoint of extreme rays, this decomposition corresponds to a \textit{bipartite} subspace of dimension $\frac{1}{2}\N(\N+1)$ spanned by all Bell-pair extreme rays (cf.\ the SA subspace), and a complementary \textit{multipartite} subspace spanned by all other extreme rays.

\section{Correlation Measures from HEIs?}
\label{sec:correlations}

Mutual information, corresponding to the total amount of correlation between two subsystems, is a prototypical correlation measure.  
It is non-negative 
$\mxI{A:B} \geq 0$ and vanishes when bipartite correlation is absent, i.e., when the density matrix is in a product state $\rhox{AB} = \rhox{A} \otimes \rhox{B} $.
It also satisfies the monotonicity property $\mxI{A:BC} \geq \mxI{A:B}$ by virtue of strong subadditivity.  

We wish to abstract away these properties for more general measures.  In particular, we require that any putative correlation measure cannot be negative and cannot increase as some part of the system is discarded.  Therefore we posit the following as basic properties of any correlation measure, viewed as a measure of the amount of some type of correlation present in a system:
\begin{enumerate} 
\item 
It should be a non-negative quantity and vanish when the correlation measured is absent.
\item 
It should be monotonically non-increasing under partial tracing.
\end{enumerate} 

Given that HEIs define non-negative information quantities on holographic states, it is natural to ask whether they can be viewed as multipartite correlation measures in holography. 
While such quantities are unlikely to satisfy the monotonicity property for general quantum states, one might hope that they are monotonic when restricted to holographic states. 
If so, this would elucidate some aspects of the structural form of the HEC.  The main result of the present section is that this not the case. By \cref{thm:superbalhec}, apart from the mutual information, all information quantities corresponding to facets of the HEC are superbalanced. 
We now show that no superbalanced information quantity can be monotonic, even with the restriction to holographic states.

\begin{nthm}
   Superbalanced information quantities are non-monotonic in holography.\footnote{ \, 
    It should not be surprising that superbalanced information quantities are non-monotonic for general quantum states, since these do not even obey any HEIs other than SA.
    The statement here is strictly stronger, and says that monotonicity is not respected even by the restricted class of holographic states.
   }
\end{nthm}

\begin{proof} 
Since here we discuss the effect of tracing out some subsystem, it will be convenient to generalize our framework to work with a mixed state on $\nN+1$ parties, where the $(\nN+1)^\text{st}$ party will be traced out.  (Hence, for this proof only, the purifier will be left implicit.)
Consider a mixed $(\N+1)$-party system with $[\N+1] = \{\x{A}_1, \cdots, \x{A}_{\N}, \x{a}_{\N}\}$. We would like to show that any $\N$-party superbalanced information quantity cannot be monotonic under tracing out $\x{a}_{\N}$, i.e., for a generic superbalanced quantity $Q\prn{\x{A}_1: \cdots :\x{A}_\N}$,
the quantity
\begin{equation}
    \Delta Q \equiv 
     Q\prn{\x{A}_1: \cdots :\x{A}_\N~\x{a}_\N} - Q\prn{\x{A}_1: \cdots :\x{A}_\N }
\end{equation}
cannot be non-negative for all states. By~\cref{prop:balsuperbal}, we can expand the superbalanced quantity $Q\prn{\x{A}_1: \cdots :\x{A}_\N}$ in terms of singleton $\x{I}_n$'s with $n \geq 3$:
\begin{equation}
    Q\prn{\x{A}_1: \cdots :\x{A}_\N } = \sum_{\pI \subseteq[N-1]\,:\,\abs{\pI}\ge2} \nu_{\pI} \, \x{I}_{|\pI|+1}(\aset{\pI}:\x{A}_\N) + (\text{terms independent of } \x{A}_\N),
\end{equation}
where 
$[\N-1] = \{\x{A}_1, \cdots , \x{A}_{\N-1}\}$. Using this expansion, the change $\Delta Q$ reads
\begin{align}
\label{eq:Q-variation}
    \Delta Q &= \sum_{\pI \subseteq[N-1]\,:\,\abs{\pI}\ge2} \nu_{\pI} \, \x{I}_{|\pI|+1}(\aset{\pI}:\x{a}_\N|\x{A}_\N) \nonumber\\
    &= \sum_{\pI \subseteq[N-1]\,:\,\abs{\pI}\ge2} \nu_{\pI} \, \prn{\x{I}_{|\pI|+1}(\aset{\pI}:\x{a}_\N) - \x{I}_{|\pI|+2}(\aset{\pI}:\x{a}_\N:\x{A}_\N)},
\end{align}
where we have used~\cref{eq:CondMultInf}. We now show that requiring $\Delta Q$ to be non-negative on every state forces all $\nu_\pI$ to vanish, by sequentially choosing states that collapse the sum onto a single term, which is non-sign-definite unless its coefficient vanishes. To achieve this, consider a class of $(\N+1)$-party mixed states $\{\rho^{(\pI)}_{[\N+1]}\}_{\pI \subseteq [\N-1]}$ with the following factorization structure
\begin{equation}
    \rho^{(\pI)}_{[\N+1]} = \pi_{\pI \x{A}_\N \x{a}_\N} \otimes \sigma_{[\N +1] \setminus \{\pI \x{A}_\N \x{a}_\N\} }
\end{equation}
with $\pi, \sigma$ being arbitrary states.
Then for any distinct pair $\pI,\pI^\prime \subseteq[\N-1]$ with $\pI^\prime \not\subseteq \pI$,
\begin{equation}
\label{eq:zeroI}
    \x{I}_{|\pI^\prime|+1}(\aset{\pI}^\prime : \x{a}_\N) \big|_{\rho^{(\pI)}} = \x{I}_{|\pI^\prime|+2}(\aset{\pI}^\prime : \x{A}_\N : \x{a}_\N) \big|_{\rho^{(\pI)}} = 0
\end{equation}
due to~\cref{lem:Ivashprod}. Now, to collapse the sum in~\cref{eq:Q-variation} to a single term, pick any $\pI$ with lowest cardinality in the sum such that $\pI^\prime \not\subseteq \pI$ for any other $\pI^\prime$ in the sum; it follows by~\cref{eq:zeroI} that
\begin{equation}
    \Delta Q \big|_{\rho^{(\pI)}} = \nu_{\pI} \, \x{I}_{|\pI|+1}(\aset{\pI}:\x{a}_\N|\x{A}_\N).
\end{equation}
Then, as we imposed no restriction on the state $\pi_{\pI \x{A}_\N \x{a}_\N}$, requiring $\Delta Q$ to be non-negative on all states forces, for this particular $\pI$, $\nu_\pI = 0$ by virtue of~\cref{lem:condInosign}. The same argument can then be applied to all other $\pI$'s with the lowest cardinality, thereby setting $\nu_\pI = 0$ for all lowest-cardinality-$\pI$'s in the sum. Then the argument can be applied to $\pI$'s with the lowest cardinality among the surviving terms in the sum. Iterating this argument sets $\nu_\pI = 0$ for all $\pI$'s, thereby proving the claim.
\end{proof}

This result shows that, besides the mutual information, none of the information quantities defined through HEIs admit an interpretation as a correlation measure in holography.

\section{New \texorpdfstring{$\N=6$}{N=6} HEIs and \ctform}
\label{ss:HEIrecasting}

Having ascertained that the higher HEIs are not natural correlation measures, the task of understanding the physical significance of their form still remains.  Often the operational meaning of a given quantity becomes more apparent when expressed in a better representation.  
We motivate a specific form for such a representation in \cref{ss:ctform.motiv}, present a set of HEIs in this form (cf.\ \cref{tab:HEI_N5} for $\N=5$ quantities and \cref{tab:HEI_N6} for  selected $\N=6$ quantities) in \cref{ss:presentHEIs}, and offer some observations regarding their structure in \cref{ss:structure}.

\subsection{Motivation for the \ctform}
\label{ss:ctform.motiv}
We have already seen in \cite{He:2019ttu} that the I-basis representation is more convenient than the S-basis one for producing more compact expressions. Here we supplement this with some observations which allow us to propose a rather rigid structural form for information quantities which we speculate may be able to capture all HEIs.

We start by noting that there are structural similarities between distinct HEI orbits and across different $\nN$.  One simple example of this is the MMI itself: while there is only a single orbit for $\nN=3$, there are two distinct orbits for $\nN=4,5$, one of which 
($-\txI{A:B:C}\ge 0$) 
can be viewed as a direct uplift from $\nN=3$, while the other ($-\txI{A:B:CD}\ge 0$) may be understood as a fine-graining thereof.\footnote{ \, 
	At $\N=5$, $\txI{A:BC:DE}$ can be viewed as a purification and permutation of $\txI{A:B:CD}$, while $\txI{A:B:CDE}$ is similarly related to $\txI{A:B:C}$.  On the other hand, at $\N=6$, there are 3 distinct orbits, generated by $\txI{A:B:C}$, $\txI{A:B:CD}$, and $\txI{A:BC:DE}$.} 
This motivates the consideration of multipartite information with polychromatic arguments to express HEIs.  Moreover, we have seen that higher $I_n$'s can be expressed in terms of smaller $n$ ones by increasing the argument size.  In principle, this means we could express all HEIs in terms of just mutual informations (with non-sign-definite coefficients), but such a recasting would manifest neither superbalance nor the I-basis sign-alternation required by \cref{prop:relcoefbases}. 
On the other hand, we might expect an efficient rewriting of HEIs to have these basic properties automatically built in. 

To guarantee superbalance, by \cref{prop:balsuperbal} it suffices to restrict the use of polychromatic $I_n$'s to $n\ge3$, while by iterating \cref{eq:Inadditivitysplit}, it is easy to check that the sign-alternation from \cref{prop:relcoefbases} can be accomplished by demanding that any $I_n$, polychromatic or not, come with a coefficient sign $(-1)^n$, as extensively exemplified in \cref{asec:heis}.  Moreover, because $I_3$ is the only sign-definite multipartite information quantity by MMI and \cref{lem:Inosign}, we also expect it to play a privileged role. 
Expressing information quantities only in terms of polychromatic $I_3$'s would however fail to produce non-redundant HEIs; we need to supplement sign-definite terms with some non-sign-definite terms.
We empirically observe that $I_4$'s and $I_3$'s appearing in HEIs naturally pair up to produce conditional (polychromatic) $I_3$'s,
and that such fortuitous groupings in fact extend to higher $n$ as well.
This motivates us to try expressing HEIs purely as combinations of $I_3$'s and conditional $I_3$'s, the latter being individually non-sign-definite by \cref{lem:condInosign} as desired.\footnote{ \, 
	Of course, one could use \cref{eq:CondMultInf} to rewrite 
	$I_4$ as a difference between $I_3$ and conditional $I_3$ in any information quantity, and similarly for higher $I_n$, but in general this would spoil the automatic $\sI_n$ sign alternation, so it is non-trivial that such wrong-sign terms cancel out in the full expression for the HEIs as required by \cref{prop:relcoefbases}.}
With these building blocks, it is also nicely straightforward to implement the desired sign-alternation property, by requiring that not just the  $I_3$'s, but also their conditional versions, must all come with a negative sign, independently of the argument size.
This motivates the following ansatz for the form of information quantities corresponding to HEIs, which we dub the \emph{\ctform}.
\begin{ndefi}[\ctform]
\label{def:ctform}
    An information quantity $Q$ is said to be in the \emph{\ctform} if it is expressed as \emph{
	\begin{equation}\label{eq:ctform}
		Q = \sum_i - I_3(\x{X}_i : \x{Y}_i : \x{Z}_i \, | \, \x{W}_i )
	\end{equation}
	}%
	where the arguments \emph{$\x{X}_i , \x{Y}_i , \x{Z}_i , \x{W}_i \subset [\nN]$} are disjoint subsystems, the sum runs over any finite number of terms, and we allow for the conditioning to trivialize, \emph{$\x{W}_i = \varnothing$}, in which case \emph{$ I_3(\x{X}_i : \x{Y}_i : \x{Z}_i \, | \, \varnothing ) \coloneqq  I_3(\x{X}_i : \x{Y}_i : \x{Z}_i )$}.\footnote{ \, Note that the \ctform\ can account for arbitrary negative integer coefficients by simply repeating terms, though in practice they turn out to be distinct for vast majority of the HEIs we have found.
}
\end{ndefi}
Each term in \cref{eq:ctform} can be characterized as either tripartite information (when $\x{W}_i = \varnothing$) or conditional tripartite information (when $\x{W}_i \ne \varnothing$), which we will abbreviate as I and C, respectively.  It will then be convenient to summarize the structural form of the HEI by a ``word'' made of I's and C's, one letter for each term, such as I, IC, IIC, etc..  The requirement of non-redundancy dictates that apart from SA and MMI, all higher HEIs must consist of at least two terms (since the single term C by itself is not sign-definite) and that at least one of the terms must be a C (because I itself is sign-definite, so multiple I's yield a redundant quantity).
The minimal such expression would then be IC, i.e.\ consist of a single $I_3$ term and a single conditional $I_3$ term.  And indeed, we find that there are HEIs of this form; 
cf.\ $\Qnn{5}{3, 2}{5}$ of \cref{tab:HEI_N5} and $\Qnn{6}{5, 5, 1}{12}$ in \cref{tab:HEI_N6}.

This is encouraging, and inspires a systematic method of looking for new HEIs:  posit the requisite \ctform, by specifying how many I's and C's a given information quantity should have, and then test all possible combinations of arguments $\{\x{X}_i , \x{Y}_i , \x{Z}_i , \x{W}_i\}$ (up to permutation symmetry between $\{\x{X}_i , \x{Y}_i , \x{Z}_i \}$ for each $i$ and overall permutation and purification symmetry).
In \cref{asec:n6heis} we elaborate on how we accomplished this search exhaustively to find all HEIs which admit a \ctform\ with four or fewer terms; there are $330$ such HEIs, listed explicitly in the ancillary files, with a representative set exemplified below in \cref{tab:HEI_N6}.\footnote{ \, Of these, $9$ are among the $11$ lifts from $\nN \le 5$; 
the other two lifts are SA (which is not superbalanced) and the last inequality in \cref{tab:HEI_N5}, which requires more than four terms in \ctform.}

Note, however, that unlike the I-basis or S-basis representation, the \ctform\ representation is highly non-unique.  A given information quantity can admit many distinct \ctform\ expressions, not only within a single I..C..\ type but also for multiple types.  This has the advantage of using different \ctform s  of the same HEI to manifest structural similarities to other HEIs.  In \cref{asec:heis} we explain how, starting from the I-basis representation, we can recast them in a \ctform.\footnote{ \,
	This is in fact a very useful exercise because it may turn out that the originally-generated form is not the most convenient one for some purposes, so we can uniquely re-express it in the I-basis and then convert back to a more convenient form.} 
 
Conversely, we emphasize that it is a priori unclear whether every HEI admits a \ctform.
In particular, the structure we propose with our \ctform\ may be non-trivially rigid: it suffices, but is not necessary, to implement superbalance and the sign-alternation property of \cref{prop:relcoefbases}. However, there are several encouraging pieces of evidence in favor of our hypothesis, which we revisit in \cref{sec:discuss}.

\subsection{Presentation of HEIs}
\label{ss:presentHEIs}
In this short subsection, we explain the HEI tables \ref{tab:HEI_N5} and \ref{tab:HEI_N6}.
By permutation symmetry, it suffices to list one instance from each symmetry orbit; since different choices may be convenient for different purposes, the given instance chosen for this section need not coincide with that listed in the ancillary file.  On the other hand, it will be instructive to express each chosen instance in several alternative forms.  In the following tables, we provide three separate forms: the conventional S-basis form, the I-basis representation, and the \ctform\ motivated in the previous paragraph.%
\footnote{ \, 
	We have also explored the corresponding quantities in the K-basis of \cite{He:2019ttu}, which manifests the larger symmetry $Sym_{\N+1}$ including purifications, but since this representation did not yield substantially new insights, we omit it in the present summary.
}
The HEI information quantities themselves are labeled by $Q$ with the superscript indicating the position number of the HEI quantity (possibly permuted) as listed in the ancillary file (see \cref{asec:n6heis}), 
and  the (optional) subscript consisting of the set of numbers $\Ink{i_3,i_4,\ldots}$, with $i_n$ characterizing the number of $\sI_n$'s when represented in the I-basis.\footnote{ \, 
This turns out to be a useful characterization, which we dub the i\# classification and utilize in \cref{asec:heis}; cf.\ \cref{tab:kclassif} for the i\# classification for the individual building blocks.
Note that for notational compactness, in this section we list only the non-zero values of $i_n$.  This is unambiguous because for any  \ctform\ expression written in the I-basis, the absence of any $\sI_m$ terms implies the absence of all $
\sI_n$ terms with $n\ge m$.
}
This decoration is mainly for convenience of providing immediate information about the quantity; when this information is not relevant for the discussion we omit it and simply quote the number $n$ which is the unique identifier of the given orbit.\footnote{ \, \label{fn:inumberclass}
	The associated value of $\N$ is in fact uniquely determined from $n$.  For $n=1$, $N=2$; for $n=2,3,4$, $\N=3$; for $5 \le n \le 11$, $\N=5$; and for $n\ge 12$, $\N=6$.  
	On the other hand, there are several possible i\# forms for the given orbit, due to the purification symmetry of $I_3$ implied by \cref{lem:Ifullsystem} (cf.\ \cref{eq:I3swap}).  Hence its use allows for a further refinement of the specification, effectively fine-graining the equivalence class to $Sym_{\N}$ sub-orbit within the given $Sym_{\N+1}$ orbit.
}
We employ the $Q$ label to indicate the full equivalence class of instances along the permutation and purification orbit, so when we write $Q=\ldots$, we mean that there is an instance along the $Q$ orbit which equals the given expression.
Below the label, for convenience we also indicate the I..C.. form.
For each HEI, we separate out the positive and negative entropy terms in the S-basis (the signs have no correlation with the subsystem size, but within each set we further order by subsystem size) and by $n$ in the I-basis (with sign $(-1)^n$), and finally within each group lexicographically.




%
\begin{table}[htbp] 
\begin{center}
\scriptsize
\begin{tabular}{| c | c | }
\hline
label &  $\nN=5$ HEI information quantities \\ 
\hline
\hline
$\QNknn{5}{3, 2}{5}$
	&  \makecell{$  -\xS{ABCD}-\xS{ACDE}-\xS{AB}-\xS{AD}-\xS{DE}-\xS{C}  $
		\\$ +\, \xS{ABC}+\xS{ABD}+\xS{ACD}+\xS{ADE}+\xS{CDE} $}
\\ 
	& \shadeI{$ \xI{ABCD}+\xI{ACDE}-\xI{ACD}-\xI{ACE}-\xI{BCD} $} 
\\ 
IC 	& \shadeIs{$ -\tI{AB:C:D}-\ctI{A:C:E}{D} $}
\\  \hline 
$\QNknn{5}{4, 2}{7}$  
	&  \makecell{$  -\xS{ABCD}-\xS{BCDE}-\xS{ABE}-\xS{BC}-\xS{BD}-\xS{A}-\xS{C}-\xS{D}-\xS{E}  $
		\\$ +\, \xS{ABC}+\xS{ABD}+\xS{BCD}+\xS{BCE}+\xS{BDE}+\xS{AE}+\xS{CD} $}
\\ 
	& \shadeI{$ \xI{ABCD}+\xI{BCDE}-\xI{ABE}-\xI{ACD}-\xI{BCD}-\xI{CDE} $} 
\\ 
IIC	& \shadeIs{$ -\tI{AB:C:D}-\tI{A:B:E}-\ctI{C:D:E}{B} $}
\\  \hline 
$\QNknn{5}{4, 3}{9}$  
	&  \makecell{$  -\xS{ABCD}-\xS{ACDE}-\xS{BCDE}-\xS{AB}-\xS{AD}-\xS{BC}-\xS{CD}-\xS{CE}-\xS{DE}  $
		\\$ +\, \xS{ABC}+\xS{ABD}+\xS{ACD}+\xS{ADE}+\xS{BCD}+\xS{BCE}+2\xS{CDE} $}
\\ 
	& \shadeI{$ \xI{ABCD}+\xI{ACDE}+\xI{BCDE}-\xI{ACD}-\xI{ACE}-\xI{BCD}-\xI{BDE} $} 
\\ 
ICC	& \shadeIs{$ -\tI{AB:C:D}-\ctI{B:D:E}{C}-\ctI{A:C:E}{D} $}
\\  \hline 
$\QNknn{5}{4, 3}{10}$  
	&  \makecell{$  -\xS{ABCE}-\xS{ACDE}-\xS{BCDE}-\xS{ABD}-\xS{AC}-\xS{AE}-\xS{BC}-\xS{BE}-\xS{CD}-\xS{DE}  $
		\\$ +\, \xS{ABC}+\xS{ABE}+\xS{ACD}+\xS{ACE}+\xS{ADE}+\xS{BCD}+\xS{BCE}+\xS{BDE}+\xS{CDE} $}
\\ 
	& \shadeI{$ \xI{ABCE}+\xI{ACDE}+\xI{BCDE}-\xI{ABD}-\xI{ACE}-\xI{BCE}-\xI{CDE} $} 
\\ 
ICCC	& \shadeIs{$ -\tI{A:B:D}-\ctI{B:C:E}{A}-\ctI{C:D:E}{B}-\ctI{A:C:E}{D} $}
\\  \hline 
$\QNknn{5}{8, 6}{11}$ 
	&  \makecell{$  -2\xS{ABCD}-\xS{ABCE}-2\xS{ABDE}-\xS{ACDE}-2\xS{AB}-\xS{AC}-2\xS{AD}-\xS{AE}-2\xS{BC}-\xS{BD}-\xS{CE}-2\xS{DE}  $
		\\$ +\, 3\xS{ABC}+3\xS{ABD}+\xS{ABE}+\xS{ACD}+\xS{ACE}+3\xS{ADE}+\xS{BCD}+\xS{BCE}+\xS{BDE}+\xS{CDE} $}
\\ 
	& \shadeI{$ 2\xI{ABCD}+\xI{ABCE}+2\xI{ABDE}+\xI{ACDE}-\xI{ABD}-2\xI{ABE}-2\xI{ACD}-\xI{ACE}-\xI{BCD}-\xI{BDE} $} 
\\ 
IICCCC	& \shadeIs{$ -\tI{AB:C:D}-\tI{AE:B:D}-\ctI{A:B:E}{C}-\ctI{A:B:E}{D}-\ctI{A:C:D}{B}-\ctI{A:C:E}{D} $}
\\  \hline 

\end{tabular}
\end{center}
\caption{HEI information quantities for $\nN=5$ in the S-basis (white background), I-basis (blue background), and the \ctform\ (pink background, with the specific I..C..\ class indicated in the left column). We label the information quantities in the form $\QNknn{\N}{i\#}{n}$, where $n$ is the identifier in the ancillary files, and the $\Ink{i\#}$ lists the non-zero entries of $\Ink{\#I_3, \#I_4, \#I_5, \#I_6}$.
The notational shorthand, explained in \cref{ss:notation}, is such that e.g.\ $\xS{ABCD} \coloneqq \S(\x{ABCD})$ and $\xI{ABCD} \coloneqq \sI_4(\x{A:B:C:D})$.
}
\label{tab:HEI_N5}
\end{table}



%
\begin{table}[htbp] 
\begin{center}
\scriptsize
\setlength\tabcolsep{4pt}
\begin{tabular}{| c | c | }
\hline
label & Selected $\nN=6$ HEI information quantities \\ 
\hline
\hline
$\QNknn{6}{5, 5, 1}{12}$
	&  \makecell{$  -\xS{ABDEF}-\xS{ABCF}-\xS{AD}-\xS{AF}-\xS{BE}-\xS{BF}-\xS{CF}  $
		\\$ +\, \xS{ABDE}+\xS{ABF}+\xS{ACF}+\xS{ADF}+\xS{BCF}+\xS{BEF} $}
\\ 
	& \shadeI{$ -\xI{ABDEF}+\xI{ABCF}+\xI{ABDF}+\xI{ABEF}+\xI{ADEF}+\xI{BDEF}-\xI{ABC}-\xI{ABF}-\xI{AEF}-\xI{BDF}-\xI{DEF} $} 
\\ 
IC		& \shadeIs{$ -\tI{AD:BE:F}-\ctI{A:B:C}{F} $}
\\  \hline 
$\QNknn{6}{4, 3}{24}$ 
	&  \makecell{$  -\xS{ABCF}-\xS{ABDF}-\xS{ABEF}-\xS{AB}-\xS{AD}-\xS{AF}-\xS{BE}-\xS{BF}-\xS{CF}  $
		\\$ +\, \xS{ABD}+\xS{ABE}+2\xS{ABF}+\xS{ACF}+\xS{ADF}+\xS{BCF}+\xS{BEF} $}
\\ 
	& \shadeI{$ \xI{ABCF}+\xI{ABDF}+\xI{ABEF}-\xI{ABC}-\xI{ABF}-\xI{AEF}-\xI{BDF} $} 
\\ 
ICC	& \shadeIs{$ -\tI{AD:B:F}-\ctI{A:E:F}{B}-\ctI{A:B:C}{F} $}
\\  \hline 
$\QNknn{6}{5, 3}{13}$
	&  \makecell{$  -\xS{ABCF}-\xS{ABEF}-\xS{CDEF}-\xS{AF}-\xS{BF}-\xS{CD}-\xS{CF}-\xS{A}-\xS{B}-\xS{E}  $
		\\$ +\, \xS{ABF}+\xS{ACF}+\xS{AEF}+\xS{BCF}+\xS{BEF}+\xS{CDE}+\xS{CDF}+\xS{AB} $}
\\ 
	& \shadeI{$ \xI{ABCF}+\xI{ABEF}+\xI{CDEF}-\xI{ABC}-\xI{ABE}-\xI{ABF}-\xI{CEF}-\xI{DEF} $} 
\\ 
IIC	& \shadeIs{$ -\tI{CD:E:F}-\tI{A:B:EF}-\ctI{A:B:C}{F} $}
\\  \hline 
$\QNknn{6}{6, 5, 1}{14}$
	&  \makecell{$  -\xS{ABDEF}-\xS{ABCF}-\xS{CEF}-\xS{AD}-\xS{AF}-\xS{BF}-\xS{B}-\xS{C}-\xS{E}  $
		\\$ +\, \xS{ADEF}+\xS{ABD}+\xS{ABF}+\xS{ACF}+\xS{BCF}+\xS{BEF}+\xS{CE} $}
\\ 
	& \shadeI{$ -\xI{ABDEF}+\xI{ABCF}+\xI{ABDE}+\xI{ABDF}+\xI{ABEF}+\xI{BDEF}-\xI{ABC}-\xI{ABE}-\xI{ABF}-\xI{BDE}-\xI{BDF}-\xI{CEF} $} 
\\ 
IIC	& \shadeIs{$ -\tI{C:E:F}-\tI{AD:B:EF}-\ctI{A:B:C}{F} $}
\\  \hline 
$\QNknn{6}{7, 6, 1}{17}$ 
	&  \makecell{$  -\xS{ABDEF}-\xS{ABCF}-\xS{CDEF}-\xS{AD}-\xS{AF}-\xS{BF}-\xS{CD}-\xS{CF}-\xS{B}-\xS{E}  $
		\\$ +\, \xS{ADEF}+\xS{ABD}+\xS{ABF}+\xS{ACF}+\xS{BCF}+\xS{BEF}+\xS{CDE}+\xS{CDF} $}
\\ 
	& \shadeI{$ -\xI{ABDEF}+\xI{ABCF}+\xI{ABDE}+\xI{ABDF}+\xI{ABEF}+\xI{BDEF}+\xI{CDEF}$} 
\\	& \shadeI{$-\xI{ABC}-\xI{ABE}-\xI{ABF}-\xI{BDE}-\xI{BDF}-\xI{CEF}-\xI{DEF} $} 
\\ 
IIC	& \shadeIs{$ -\tI{CD:E:F}-\tI{AD:B:EF}-\ctI{A:B:C}{F} $}
\\  \hline 
$\QNknn{6}{6, 4}{125}$ 
	&  \makecell{$  -\xS{ABCF}-\xS{ABDE}-\xS{ABEF}-\xS{CDEF}-\xS{AE}-\xS{AF}-\xS{BE}-\xS{BF}-\xS{CD}-\xS{CF}-\xS{DE}-\xS{A}-\xS{B}  $
		\\$ +\, \xS{ABE}+\xS{ABF}+\xS{ACF}+\xS{ADE}+\xS{AEF}+\xS{BCF}+\xS{BDE}+\xS{BEF}+\xS{CDE}+\xS{CDF}+\xS{AB} $}
\\ 
	& \shadeI{$ \xI{ABCF}+\xI{ABDE}+\xI{ABEF}+\xI{CDEF}-\xI{ABC}-\xI{ABD}-\xI{ABE}-\xI{ABF}-\xI{CEF}-\xI{DEF} $} 
\\ 
IICC		& \shadeIs{$ -\tI{CD:E:F}-\tI{A:B:EF}-\ctI{A:B:D}{E}-\ctI{A:B:C}{F} $}
\\  \hline 
$\QNknn{6}{8,8,2}{19}$  
	&  \makecell{$  -\xS{ABCDF}-\xS{ABDEF}-\xS{BCDEF}-\xS{BCE}-\xS{CDF}-\xS{AD}-\xS{AF}-\xS{BE}-\xS{BF}-\xS{D}-\xS{F}  $
		\\$ +\, \xS{ABDE}+\xS{ACDF}+\xS{BCDE}+\xS{BCDF}+\xS{BCEF}+\xS{ABF}+\xS{ADF}+\xS{BEF}+\xS{DF} $}
\\ 
	& \shadeI{$ -\xI{ABCDF}-\xI{ABDEF}-\xI{BCDEF}+\xI{ABCD}+\xI{ABCF}+2\xI{ABDF}+\xI{ABEF}+\xI{ADEF}+\xI{BCDF}+2\xI{BDEF}$} 
\\	& \shadeI{$+\xI{CDEF}-\xI{ABC}-\xI{ABD}-\xI{ABF}-\xI{AEF}-2\xI{BDF}-\xI{CDF}-2\xI{DEF} $} 
\\ 
IIC	& \shadeIs{$ -\tI{BCE:D:F}-\tI{AD:BE:F}-\ctI{A:B:CD}{F} $}
\\  \hline 
$\QNknn{6}{8, 8, 2}{21}$ 
	&  \makecell{$  -\xS{ABCDF}-\xS{ABDEF}-\xS{CDF}-\xS{AD}-\xS{AF}-\xS{BC}-\xS{BE}-\xS{BF}-\xS{D}-\xS{F}  $
		\\$ +\, \xS{ABDE}+\xS{ACDF}+\xS{ABF}+\xS{ADF}+\xS{BCD}+\xS{BCF}+\xS{BEF}+\xS{DF} $}
\\ 
	& \shadeI{$ -\xI{ABCDF}-\xI{ABDEF}+\xI{ABCD}+\xI{ABCF}+2\xI{ABDF}+\xI{ABEF}+\xI{ADEF}+\xI{BCDF}+\xI{BDEF}$} 
\\	& \shadeI{$-\xI{ABC}-\xI{ABD}-\xI{ABF}-\xI{AEF}-2\xI{BDF}-\xI{CDF}-\xI{DEF} $} 
\\ 
IIC	& \shadeIs{$ -\tI{BC:D:F}-\tI{AD:BE:F}-\ctI{A:B:CD}{F} $}
\\  \hline 
$\QNknn{6}{10, 11, 3}{23}$ 
	&  \makecell{$  -\xS{ABCDE}-\xS{ABCDF}-\xS{ABDEF}-\xS{CDF}-\xS{AD}-\xS{AF}-\xS{BC}-\xS{BE}-\xS{BF}-\xS{DE}-\xS{A}  $
		\\$ +\, \xS{ABDE}+\xS{ACDF}+\xS{BCDE}+\xS{BCDF}+\xS{ABC}+\xS{ABF}+\xS{ADE}+\xS{ADF}+\xS{BEF} $}
\\ 
	& \shadeI{$ -\xI{ABCDE}-\xI{ABCDF}-\xI{ABDEF}+2\xI{ABCD}+\xI{ABCE}+\xI{ABCF}+\xI{ABDE}+2\xI{ABDF}+\xI{ABEF}+\xI{ACDE}$} 
\\	& \shadeI{$+\xI{ADEF}+\xI{BDEF} -\xI{ABC}-2\xI{ABD}-\xI{ABE}-\xI{ABF}-\xI{ACD}-\xI{ACE}-\xI{AEF}-\xI{BDF}-\xI{DEF} $} 
\\ 
IIC	& \shadeIs{$ -\tI{A:BC:DE}-\tI{AD:BE:F}-\ctI{A:B:CD}{F} $}
\\  \hline 
$\QNknn{6}{8, 6, 1}{35}$ 
	&  \makecell{$  -\xS{ABCDF}-\xS{ACEF}-\xS{BDEF}-\xS{CDF}-2\xS{AF}-\xS{DE}-\xS{EF}-\xS{A}-\xS{B}-2\xS{C}-\xS{D}  $
		\\$ +\, \xS{BCDF}+\xS{ABF}+\xS{ACF}+\xS{ADF}+\xS{AEF}+\xS{BDE}+\xS{CEF}+\xS{DEF}+\xS{AC}+\xS{CD} $}
\\ 
	& \shadeI{$ -\xI{ABCDF}+\xI{ABCD}+\xI{ABCF}+\xI{ABDF}+\xI{ACDF}+\xI{ACEF}+\xI{BDEF}$} 
\\	& \shadeI{$-\xI{ABC}-\xI{ABD}-\xI{ACD}-\xI{ACE}-\xI{ACF}-\xI{BDF}-\xI{BEF}-\xI{CDF} $} 
\\ 
IIIC	& \shadeIs{$ -\tI{AF:C:D}-\tI{B:DE:F}-\tI{A:C:EF}-\ctI{A:B:CD}{F} $}
\\  \hline 
$\QNknn{6}{9, 8, 2}{46}$ 
	&  \makecell{$  -2\xS{ABCDF}-\xS{ACEF}-\xS{BDEF}-\xS{CDF}-\xS{AF}-\xS{DE}-\xS{EF}-\xS{A}-\xS{B}-2\xS{C}-\xS{D}  $
		\\$ +\, \xS{ABCF}+\xS{ABDF}+\xS{ACDF}+\xS{BCDF}+\xS{AEF}+\xS{BDE}+\xS{CEF}+\xS{DEF}+\xS{AC}+\xS{CD} $}
\\ 
	& \shadeI{$ -2\xI{ABCDF}+2\xI{ABCD}+\xI{ABCF}+\xI{ABDF}+\xI{ACDF}+\xI{ACEF}+\xI{BCDF}+\xI{BDEF}$} 
\\	& \shadeI{$-\xI{ABC}-\xI{ABD}-\xI{ACD}-\xI{ACE}-\xI{ACF}-\xI{BCD}-\xI{BDF}-\xI{BEF}-\xI{CDF} $} 
\\ 
IIIC	& \shadeIs{$ -\tI{ABF:C:D}-\tI{B:DE:F}-\tI{A:C:EF}-\ctI{A:B:CD}{F} $}
\\  \hline 
\end{tabular}
\end{center}
\caption{HEI information quantities for $\nN=6$, with the same notational conventions as in \cref{tab:HEI_N5}.
}
\label{tab:HEI_N6}
\end{table}
  

In the rest of this section, we make a few observations regarding the actual expressions.  These are mostly of a phenomenological nature; we relegate a more systematic analysis to future work.
\subsection{Structural properties}
\label{ss:structure}

\paragraph{Compactness:}
Evidently, the \ctform\ has many fewer terms than the I-basis or S-basis forms, typically by a factor of 5 or so, though in some cases the reduction is from 47 terms in the I-basis to only 4 in the \ctform.
A priori, this was not guaranteed by the arguments that motivated these recastings (-- one can certainly cook up information quantities for which this is not the case, or indeed those which do not admit \ctform\ in the first place, such as single $\sI_n$ for $n\ne3$), so the very fact is intriguing.  Moreover, we can see that many $\N=6$ HEIs are in fact simpler than some $\N=5$ HEIs, a feature we expect to prevail at higher $\N$. However, while the original hope was that the recasting into the \ctform\ would be sufficiently suggestive to explicate the operational meaning of these HEIs, this hope has not yet been realized.  We therefore leave it as a challenge for the future.

\paragraph{Structural similarities:}
Upon closer examination of the \ctform\ expressions for the HEIs, we find that many HEIs are very similar to each other, sometimes differing by only a single letter in a single argument.  Such structural similarities between distinct HEI orbits may suggest a common origin, in the sense of modifying the same lower-$\nN$ HEI in combinatorially distinct ways.  For example, let us compare the compact form for 
$\Qnn{5}{4, 2}{7}$, $\Qnn{6}{5, 3}{13}$, $\Qnn{6}{6, 5, 1}{14}$, and $\Qnn{6}{7, 6, 1}{17}$ (the first from \cref{tab:HEI_N5} and the other three from \cref{tab:HEI_N6}, for convenience reproduced here):
\begin{equation}\label{eq:compact_compare}
	\mqty{ 
		\QNknn{5}{4, 2}{7} &=& -\tI{C:E:F}-\tI{A:B:EF}-\ctI{A:B:C}{F} \\
		\QNknn{6}{5, 3}{13} &=& -\tI{CD:E:F}-\tI{A:B:EF}-\ctI{A:B:C}{F} \\
		\QNknn{6}{6, 5, 1}{14} &=& -\tI{C:E:F}-\tI{AD:B:EF}-\ctI{A:B:C}{F} \\
		\QNknn{6}{7, 6, 1}{17} &=& -\tI{CD:E:F}-\tI{AD:B:EF}-\ctI{A:B:C}{F}
		}
\end{equation}
We see that the last three ($\N=6$) quantities build up on the first ($\N=5$) one by augmenting a single argument in, respectively, the first, second, and both, $I_3$ terms, by a new party $\x{D}$ which did not appear in the $\N=5$ quantity.  Since the two $I_3$ terms in the $\N=5$ quantity ($\Qnn{5}{4, 2}{7}$) 
had arguments with different cardinalities, all four quantities have different number of terms when expanded in the I-basis (as indicated by the i\# subscripts).  Hence the structural relation between these four expressions would have been difficult to discover without the compact form.  Moreover, the originally obtained quantities might be in different orbits, making their relations even more obscure.  On the other hand, in the \ctform\ with only a few terms, it is much easier to obtain a permutation that aligns the quantities for direct comparison as above, which is particularly convenient when the arguments have different cardinalities.\footnote{ \,
	Indeed, the specific instance of each permutation orbit was chosen precisely so as to render the compact expressions conveniently aligned.  In particular, we condition on $\x{F}$, with the last term always being $\ctI{A:B:C}{F}$, and the remaining freedom in permuting is used to obtain as much similarity as possible on the other terms.
}

Another type of structural relation is illustrated by the following triplet:
\begin{equation}\label{eq:compact_compadd2}
	\mqty{ 
		\QNknn{6}{8, 6, 1}{35} 
            &=& -\tI{AF:C:D}-\tI{B:DE:F}-\tI{A:C:EF}-\ctI{A:B:CD}{F} \\
		\QNknn{6}{8, 6, 1}{33} 
            &=& -\tI{BF:C:D}-\tI{B:DE:F}-\tI{A:C:EF}-\ctI{A:B:CD}{F} \\
		\QNknn{6}{9, 8, 2}{46} 
            &=& -\tI{ABF:C:D}-\tI{B:DE:F}-\tI{A:C:EF}-\ctI{A:B:CD}{F}
		}
\end{equation}
The only difference between these three expressions is in the first argument of the first term, where $\x{F}$ is augmented by $\x{A}$, $\x{B}$, and $\x{AB}$, respectively.\footnote{ \, 
	Note that the first two expressions, while having the same number of I-basis terms, belong to distinct orbits, since $\x{A}$ and $\x{B}$ play different roles in the remaining terms.}
On the other hand, this type of `merging' of a pair of HEIs does not always produce the third.  For example, while there is a structurally analogous companion to $\QNknn{6}{7, 6, 1}{17}$ 
(cf.\ \cref{eq:compact_compare}), 
where the first argument of the first term changes from $\x{CD}$ to $\x{AC}$,
the `merged' version,  
with the first argument becoming $\x{ACD}$,
is not an HEI.\footnote{ \, 
In this case it is a valid holographic inequality, but it does not define a facet, because it is redundant (in fact it is the sum of two different permutations of an uplift of the simplest $\nN=5$ HEI, $\QNknn{5}{3, 2}{5}$).}

There are several other types of near-complete similarities between certain pairs of HEIs.  One rather common type is when one singleton gets `transferred' between two arguments in a single I term, such as
\begin{equation}\label{eq:compact_compadd3}
	\mqty{ 
		\QNknn{6}{9, 9, 2}{219} 
            &=& -\tI{BF:C:D}-\tI{AC:B:DE}-\ctI{A:B:C}{D}-\ctI{A:CE:F}{B} \\
		\QNknn{6}{9, 9, 2}{222} 
            &=& -\tI{B:CF:D}-\tI{AC:B:DE}-\ctI{A:B:C}{D}-\ctI{A:CE:F}{B}
		}
\end{equation}
where in the 1st term, the $\x{F}$ gets transferred from $\x{B}$ to $\x{C}$.  Often, but not always, this in fact constitutes a pair of augmentations; in this case, the expression without the $\x{F}$ on either argument in the 1st term is indeed another HEI, namely $\Qnn{6}{8,8,2}{19}$. 
It would be interesting to understand such structural relations more systematically; we leave this for future explorations.

\paragraph{Building up HEIs:}
One particularly useful outcome of a systematic understanding of the structural relations between the HEIs would be a more direct solution-generating technique for higher-$\N$ HEIs, by bootstrapping successive augmentations or other modifications.  To exemplify the idea, consider MMI (uplifted from $\N=3$),
$\QNknn{3}{1}{2} = -\xI{A:B:C}$.  We can augment one argument, say $\x{C} \to \x{CD}$, to get a fine-grained version of MMI,\footnote{ \, 
	Equivalently, we can obtain this by considering the original MMI at $\N=4$, using \cref{eq:I3swap} to replace one argument by complement of the union of the original ones, uplifting to $\N=6$, and permuting to get the desired form.
} which is in fact $\QNknn{4}{2, 1}{3}=-\txI{A:B:CD}$.   Using \cref{eq:CondMultInf}, this can be recast as $\QNknn{4}{2, 1}{3}=-\txI{A:B:D}-\ctxI{A:B:C}{D}$.  Now, we can view this as the starting point and augment one argument in the first term, say $\x{B} \to \x{BE}$.  (Note that this is no longer a mere fine-graining, since other occurrences of $\x{B}$ which appear elsewhere in the quantity are left un-augmented.)  Although such an augmentation was not a-priori guaranteed to lead to a valid HEI, nevertheless, it is indeed one, namely
$\QNknn{5}{3, 2}{5}=-\tI{A:BE:D}-\ctI{A:B:C}{D}$.  Since $\x{F}$ does not appear, this is intrinsically an $\N=5$ HEI.  However, we can also augment two of the $I_3$ arguments by different subsystems, resulting in a genuine $\N=6$ HEI, namely
$\QNknn{6}{5, 5, 1}{12} = -\tI{AF:BE:D}-\ctI{A:B:C}{D}$.  
Note that had we merely fine-grained $\x{A}\to\x{AF}$ in $\QNknn{5}{3, 2}{5}$, we would have obtained an uplift of $\QNknn{5}{3, 2}{5}$ which in our search for $\N=6$ HEIs indeed showed up, as another instance in the orbit of $\QNknn{5}{5, 5, 1}{6}$. 

A distinct way of building up new HEIs at a fixed $\N$ is to leave the arguments in the individual tripartite terms unmodified, but add more tripartite terms.  Non-redundancy requires that each new term is a C rather than I.  We can then obtain a chain of HEIs, such as I - IC - ICC - ICCC, which has a realization already at $\N=5$:
\begin{equation}\label{eq:compact_compadd}
	\mqty{ 
		\QNknn{5}{2, 1}{3} 
            &=& -\tI{A:BE:D}& &  & \\
		\QNknn{5}{3, 2}{5} 
            &=& -\tI{A:BE:D} &-\ctI{A:B:C}{D}&  &  \\
		\QNknn{5}{4, 3}{9} 
            &=& -\tI{A:BE:D}& -\ctI{A:B:C}{D} & -\ctI{C:D:E}{A} & \\
		\QNknn{5}{5, 5, 1}{6} 
            &=& -\tI{A:BE:D} &-\ctI{A:B:C}{D}& -\ctI{C:D:E}{A} & -\ctI{B:C:E}{AD}
		}
\end{equation}
For both types of building up new HEIs from simpler ones, the crucial question is what are the necessary and sufficient conditions for its applicability; in other words, what characterizes the augmentations which produce another HEI.

\section{Discussion}
\label{sec:discuss}

Our main result is that the non-negative information quantities associated to HEIs cannot at face value be understood as measures of correlations, as they lack the basic property of monotonicity under partial tracing. This does not preclude the existence of useful multipartite measures of holographic correlations, but does eliminate the only natural candidates arising from the polyhedral structure of the holographic entropy cone.\footnote{ \, 
Since measures of multipartite correlations are famously poorly understood, it is also conceivable that our requirement of monotonicity simply be too strong. Our result would then potentially motivate a search for weaker, non-bipartite notions of monotonicity, or for a more relational requirement on correlations among subsystems, rather than a one-directional comparison between competing terms under partial tracing.} Besides this negative result, throughout this paper we have proved several structural properties of HEIs, as well as explored a fruitful rewriting thereof in terms of $I_3$ and conditional $I_3$ quantities. 
This \ctform\ has proven to be a powerful heuristic ansatz for the discovery of a plethora of novel HEIs which constitute facets of the $\nN=6$ holographic entropy cone (see \cref{asec:n6heis}).

While we have not proved that every HEI must necessarily be expressible in \ctform, it is natural to conjecture that this is indeed the case.  One encouraging piece of evidence comes from the fact that all the known HEIs for $\N \le 6$ can be recast in this form. Even the new $\N=7$ HEI introduced by \cite{Czech:2022fzb} (cf.\ their eq.~(2.1), which has the $\sI_n$-multiplicities $\Ink{10, 15, 6}$) can likewise be recast in the \ctform, as ICCCCCC.  This seems a highly-nontrivial check, since it has the minimal number of $\sI_3$'s possible to `soak up' all the higher $\sI_n$ terms.  Indeed, \cite{Czech:2022fzb} motivated this inequality by positing ``oxidation" of $\N=5$ HEI, built so as to guarantee that upon trivializing any subsystem, the HEI necessarily remains valid.  The building blocks of the \ctform\ partly incorporate this, in the sense that when the oxidized party is conditioned on, the C term reduces to I, while if it is one of the other arguments, the term trivializes, in both cases ensuring sign-definiteness.  However this by itself does not guarantee sign-definiteness when only a subset of an argument in a C term trivializes.  It would be interesting to see whether this observation could be utilized to further constrain the possible combinations of arguments in multiple terms of the \ctform. 

The challenge of understanding and classifying HEIs remains open. Besides recasting HEIs in \ctform, one could try to extract structural properties of HEIs by exploring their symmetries. Naively characterizing the symmetry group of HEIs is however not a fruitful venue: most $\nN=6$ HEIs have no symmetries at all. Instead, a more promising pursuit consists of decomposing HEIs into representations of the symmetric group \cite{Czech:2022fzb}, an idea that generalizes the intriguing results looking at symmetric invariants of \cite{Czech:2021rxe,Fadel:2021urx}.

In a parallel series of works \cite{Hernandez-Cuenca:2019jpv,Hernandez-Cuenca:2022pst,He:2022bmi,He:2023cco,He:2023aif} which focus on extreme rays, rather than the facets, of the HEC, a more primal structure plays a prominent role:  One can define the subadditivity cone (SAC) for any number of parties as the intersection of half-spaces given by all instances of SA, cf.\ \cref{eq:SA}.  Since this region of entropy space merely restricts the entropy vector to have a non-negative amount of correlation between any pair of subsystems, it contains the HEC (strictly for all $\N \ge 3$).  Nevertheless, subject to a certain conjecture, \cite{Hernandez-Cuenca:2022pst} showed that the HEC extreme rays can be obtained from the set of holographically-realizable extreme rays of the SAC of a more fine-grained system, by a suitable coarse-graining procedure.\footnote{ \, 
	Operationally, while finding all the extreme rays of the SAC is prohibitively difficult for larger $\N$, obtaining the vastly smaller set of SSA-compatible ones can be performed much more efficiently, and in fact can be used to obtain the full set of holographically realizable extreme rays of the SAC for $\N=6$ \cite{He:2023wip}. }
In other words, subadditivity fundamentally underlies the emergence of the HEC.  It would be interesting to explore whether the relevant SAC extreme rays can be repackaged to directly determine the HEIs, perhaps as a combination of polychromatic mutual information quantities.  This would then elucidate the role of HEIs and their relation to correlation measures.

Another recurring theme, present in both facet and extreme ray representations, is the structural relation between different values of $\N$.  It was already emphasized in \cite{Hernandez-Cuenca:2022pst} that since the holographic system can be refined arbitrarily, fixing a specific value of $\N$ is artificial.  It would be useful to formulate a more algorithmic way of `bootstrapping' from smaller $\N$ to larger $\N$.  In the case of extreme rays, the graphical representation of larger-$\N$ extreme rays typically
resembles some gluing of smaller-$\N$ extreme rays.\footnote{ \, This can be observed in fig.\ 12 of \cite{Hernandez-Cuenca:2022pst}. We thank Bogdan Stoica for some early explorations and conversations about this idea.}
For facets, we have seen that they admit structural similarities when presented in the \ctform\ which persist across $\N$.  
Again, it would be interesting to examine whether these two observations are connected in some natural way.

Finally, recall that the multipartite information quantities, apart from their convenience in presenting the HEIs, also provide a useful diagnostic of the basic geometrical features of the bulk configuration.  In particular, they diagnose the connectivity of the entanglement wedges of the various subsystems, cf.\ \cref{prop:InEWconnectivity}.
This connectivity has operational implications for the HEIs.  In the case of mutual information, \cite{May:2019odp} used focusing arguments in general relativity to argue for existence of efficient non-local quantum computation protocols, which was further generalized by  \cite{May:2022clu} in ``$n$-to-$n$ connected wedge theorem''.  
In fact, the consequence of theorem 19 in \cite{May:2022clu} on the non-vanishing mutual information across any bipartition of the full system might be more naturally rephrased as simply the non-vanishing of $I_n$.
It would be interesting to see if this can provide any insight towards formulating a more operational meaning of the higher-party HEIs.

\acknowledgments

It is a pleasure to thank 
Patrick Hayden,
Matt Headrick,
Geoff Penington,
Mukund Rangamani,
Jonathan Sorce, 
and especially Max Rota
for useful discussions.
V.H.\ and S.H-C.\ thank the Centro de Ciencias de Benasque Pedro Pascual for hospitality during the program ``Gravity - New perspectives from strings and higher dimensions'', and also the Yukawa Institute for Theoretical Physics for hospitality during the long term workshop YITP-T-23-01 where this paper was completed.
V.H.\ and F.J.\ are supported in part by the U.S.\ Department of Energy grant DE-SC0020360 under the HEP-QIS QuantISED program and by funds from the University of California.
S.H-C.\ is supported by the U.S. Department of Energy Award DE-SC0021886.

\appendix

\section{Factorized notation for \texorpdfstring{$\sI_n$}{In}'s}
\label{asec:Infactornotn}

Here we briefly describe a powerful notational simplification, which serves to motivate  interesting relational properties of the multipartite informations.  Let us start with the entropy basis representation shorthand $\xS{AB} \eqr \x{AB}$ etc., where we can then rewrite a given expression in a `factorized' form, for example, $\x{AB}+\x{BC} = \x{(A+C) B}$, where the `multiplication' is to be understood purely at the level of notation (i.e., presenting the expression in a compact form) rather than pertaining to the actual values of the entropy itself; we will use the $\eqr$ sign to remind ourselves of this fact.  It will further be convenient to use the notation $\varnothing \eqr 1$, which implements the relation $\x{X} \cup \varnothing = \x{X}$.  Any residual instances of $1$ in the resulting expression then represent the empty set and can be ignored (or inserted at will).  To avoid ambiguity, we will put a tilde, $\Tilde{\x{X}} \equiv -1+\x{X}$.  This factorized form becomes particularly compact for the $\sI_n$'s, which (up to an overall sign) can be written as $n$ simple factors of $\Tilde{\x{X}}$ for each argument $\x{X}$.  For example,
\begin{equation}\label{eq:Intfactorization}
	\til{\sI}_2(\x{A}_1:\x{A}_2) \eqr -1+\x{A}_1+\x{A}_2-\x{A}_1\,\x{A}_2 = - (-1+\x{A}_1) \, (-1+\x{A}_2) 
    = - \Tilde{\x{A}}_1 \, \Tilde{\x{A}}_2 
\end{equation}	
and more generally\footnote{ \, 
Since we define $I_1(\x{A}):=\x{A}$, we can set $\til{I}_1(\x{A}):=\Tilde{\x{A}}$ and $I_0:=0$, sot that $\til{I}_0:=-1$.
}
\begin{equation}\label{eq:Infactorization}
	\til{\sI}_n(\x{A}_1:\ldots:\x{A}_n) \eqr 
	\pqty{-1}^{n+1} \,  
	\Tilde{\x{A}}_1 \cdots \Tilde{\x{A}}_n 
\end{equation}	

This shorthand factorized form allows us to prove interesting relations between the multipartite informations.
Particularly useful one is 
\begin{equation}\label{eq:Intproduct}
	\til{\sI}_{n+m}(\x{A}_1:\ldots:\x{A}_n:\x{B}_1:\ldots:\x{B}_m) = 
	-\til{\sI}_n(\x{A}_1:\ldots:\x{A}_n)  \, \til{\sI}_m(\x{B}_1:\ldots:\x{B}_m) 
\end{equation}	
for any $n,m\ge0$.  
In terms of the untilded quantities, we can rewrite this as
\begin{equation}\label{eq:Inproduct}
\begin{split}
	\sI_{n+m}&(\x{A}_1:\ldots:\x{A}_n:\x{B}_1:\ldots:\x{B}_m)  \\
	&= - \sI_n(\x{A}_1:\ldots:\x{A}_n)  \, \sI_m(\x{B}_1:\ldots:\x{B}_m) 
	+ \sI_n(\x{A}_1:\ldots:\x{A}_n) + \sI_m(\x{B}_1:\ldots:\x{B}_m) \ .
\end{split}\end{equation}	
It is worth stressing that the `multiplication' employed here does not commute with evaluating the actual entropy whereas the summation does.

\section{Obtaining \ctform\ for HEIs}
\label{asec:heis}

In \cref{ss:HEIrecasting} we have presented a number of new $\N=6$ HEIs, as well as previously known $\N=5$ HEIs, written in a \ctform.  We have seen that this compact recasting is especially convenient, not only to represent the information quantity efficiently, but also to extract structural relations with other information quantities.  
In this appendix we explain how we can obtain this compact form, starting from the I-basis representation.  (In fact this is how we originally arrived at the compact form of the $\N=5$ HEIs, which inspired the presently-utilized search strategy.)

\paragraph{Grouping $\sI_n$ terms to compact expressions:}
Consider any information quantity writ- ten out (uniquely) in the I-basis. We can often group terms (non-uniquely) to re-write it in terms of fewer terms involving $I_n$'s with polychromatic arguments or involving conditional $I_n$'s.  This will turn out to be particularly efficient for the information quantities associated to HEIs, where the final expression involves only $n=3$, i.e.\ tripartite information and conditional tripartite information.  

The conversion makes use of \cref{eq:CondMultInf}, which for $n=3$ can be written as
 \begin{equation} \label{eq:CondTripInf}
    I_3(\x{B}:\aset{\pI} \vert \x{A}) 
    =  I_3(\x{AB}:\aset{\pI}) - I_3(\x{A}:\aset{\pI}) \\
    =  I_3(\x{B}:\aset{\pI}) - I_{4}(\x{A}:\x{B}:\aset{\pI}) \ .
\end{equation}
In particular we can group the pair $\sI_3- \sI_4$ to get conditional $I_3$, 
\begin{equation}\label{eq:simpcond3}
\xI{ABC} - \xI{ABCD} = \ctxI{A:B:C}{D} \ ,
\end{equation}
or we can group the triplet $\sI_3+ \sI_3- \sI_4$ to get polychromatic $I_3$, 
\begin{equation}\label{eq:simppol3}
\xI{ABC} + \xI{ABD} - \xI{ABCD} = \txI{A:B:CD} \ .
\end{equation}

Furthermore, we can iterate to group more terms.  For example, for $I_3$ having arguments with cardinalities  $\pqty{1:1:3}$, we can write this out as the 7-term expression with 3 $\sI_3$'s, 3 $\sI_4$'s, and 1 $\sI_5$,
\begin{equation}\label{eq:simppol113}
 \xI{ABC} + \xI{ABD} + \xI{ABE} - \xI{ABCD} - \xI{ABCE} - \xI{ABDE} + \xI{ABCDE}
 = \txI{A:B:CDE} \ ,
\end{equation}
so if the I-basis expression happens to include these 7 terms, we can compress them into the single $\tI{1:1:3}$ term.
Notice that both singletons from the RHS (here $\x{A}$ and $\x{B}$) appear as arguments in each term on the LHS.
Similarly, if we have $I_3$ with argument cardinalities  $\pqty{1:2:2}$, then this gets decomposed into 9 terms, namely  4 $\sI_3$'s, 4 $\sI_4$'s, and 1 $\sI_5$, 
\begin{equation}\label{eq:simppol122}
\xI{ABD} + \xI{ABE} + \xI{ACD} + \xI{ACE} - \xI{ABCD} - \xI{ABCE} - \xI{ABDE} - \xI{ACDE} +\xI{ABCDE} =\txI{A:BC:DE} \ ,
\end{equation}
where the singleton argument on the RHS (here $\x{A}$) is the common subscript on all the $\sI_n$ terms, and the grouping for the doubletons is such that its `factorization' gives the given $\sI_3$ terms.  

In fact, this pattern continues for arbitrary cardinalities of the arguments.  We have a de-facto `factorization' between all arguments, which comes from the special structure of the $\sI_n$'s explained in \cref{asec:Infactornotn}.  Specifically, the colons separating the arguments act as multiplication, and within each argument composed of $\ell$ singletons, say $\x{X} = \prod_i \x{A}_i$
the factor is given by $[1 -\prod_i  (1-\x{A}_i))]$.  
(Note that when $\ell=1$, this factor reduces to just $\x{A}_1$.)

Moreover, this type of structure extends to the conditioning as well: there we simply replace 
$[1-\Pi]$ by $[\Pi]$, which becomes the product $\prod_j(1-\x{B}_j)$ over each singleton $\x{B}_j$ comprising the party conditioned on.  For example, to extract the expression for conditioning on a doubleton, we can combine the following 4 $\sI_n$ terms:
\begin{equation}\label{eq:simpcond2}
\xI{ABC} - \xI{ABCD} - \xI{ABCE} +\xI{ABCDE} = \ctxI{A:B:C}{DE} \ .
\end{equation}
Notice that here $\aset{\x{ABC}}$ appears in all terms, while the arguments conditioned on ($\x{CD}$) appear in the combination $(1-\x{C})(1-\x{D})$.  
Similarly, we can of course also combine the rule for polychromatic rewriting with that for conditioning rewriting, with the `factorization' combined accordingly; e.g.
\begin{equation}\label{eq:simppolcond}
\xI{ABC} + \xI{ABD} - \xI{ABCD} - \xI{ABCE} - \xI{ABDE} +\xI{ABCDE}=  \ctxI{A:B:CD}{E} \ .
\end{equation}

Conversely, one could likewise write a single $\sI_n$ in terms of polychromatic and conditional tripartite informations, but with non-sign-definite coefficients; for example:
\begin{equation} \label{eq:I5expand}
\begin{aligned}
&\xI{ABCDE} \\
	&\!=   \!\txI{A:D:E} + \!\txI{B:D:E} + \!\txI{C:D:E}  - \!\txI{AB:D:E} - \!\txI{AC:D:E} - \!\txI{BC:D:E} + \!\txI{ABC:D:E} \\
	&= \!\txI{C:D:E}- \!\ctxI{C:D:E}{A} - \!\ctxI{C:D:E}{B} +  \!\ctxI{C:D:E}{AB} \ .
\end{aligned}
\end{equation}
Notice that in the first expression, the first argument from each term groups to a form of $\xI{ABC}$ with $\aset{\x{DE}}$ in the other two arguments coming along for the ride, while in the second expression the conditioning is on parties grouped effectively as $(1-\x{A})(1-\x{B})$, now with $\aset{\x{CDE}}$ coming along for the ride.

Finally, note that  we can also rewrite the polychromatic $I_3$ in terms of conditional and monochromatic ones; for example,
\begin{equation}\label{eq:pol2}
\txI{A:B:CD} = \ctxI{A:B:C}{D} + \txI{A:B:D} 
\end{equation}
\begin{equation}\label{eq:pol3}
\txI{A:B:CDE} = \ctxI{A:B:C}{DE} + \ctxI{A:B:D}{E} + \txI{A:B:E} 
\end{equation}
and so forth,
so we could in principle reduce all polychromatic \ctform\ quantities to expressions involving just singleton $I_3$'s conditioned on polychromatic arguments (with the singleton $\sI_3$ terms regarded as conditioned on zero-cardinality subsystem).

\paragraph{Simple example, $\QNknn{6}{4,3}{24}$:} 

In the I-basis, the expression for some specific instance along the orbit is 
$$
\QNknn{6}{4,3}{24} 
	= + \xI{ABCD} + \xI{ABCE} + \xI{ABCF} - \xI{ABC} - \xI{ABF} - \xI{ACE} - \xI{BCD} 
$$
which involves 3 $\sI_4$'s and 4 $\sI_3$'s.  Examining the subscripts, we see that we can group one of the $\sI_4$'s with two of the $\sI_3$'s to get a polychromatic $I_3$ and the remaining 2 $\sI_4$'s each with a single $\sI_3$ to get conditional $I_3$.  Since $\x{ABC}$ is included in each of the three $\sI_4$ subscripts, whereas the remaining 3 $\sI_3$ subscripts are contained in precisely one of the $\sI_4$ subscripts, using \cref{eq:simpcond3} and \cref{eq:simppol3}, we have three options at the above-specified grouping:
\begin{equation} \label{eq:Q18simplif}
\begin{aligned}
\QNknn{6}{4,3}{24} 
	&= \pqty{\xI{ABCD} - \xI{ABC} - \xI{BCD}} + \pqty{\xI{ABCE}- \xI{ACE}}  + \pqty{\xI{ABCF} - \xI{ABF}} \\
	& \qquad = -\txI{B:C:AD}  -\ctxI{A:C:E}{B}-\ctxI{A:B:F}{C}\\
	&= \pqty{\xI{ABCD}  - \xI{BCD}}+ \pqty{\xI{ABCE}  - \xI{ABC} - \xI{ACE}} + \pqty{\xI{ABCF} - \xI{ABF}}\\
	& \qquad = -\ctxI{B:C:D}{A} -\txI{A:C:BE} -\ctxI{A:B:F}{C}\\
	&= \pqty{\xI{ABCD} - \xI{BCD}} + \pqty{\xI{ABCE}  - \xI{ACE}} + \pqty{\xI{ABCF} - \xI{ABC} - \xI{ABF}} \\
	& \qquad = -\ctxI{B:C:D}{A}-\ctxI{A:C:E}{B}-\txI{A:B:CF}
\end{aligned}
\end{equation}
Note that all three compact expressions have the same structural form (not just as ICC, or even more specifically as a single tripartite information with argument cardinalities $\pqty{1:1:2}$ along with two $\pqty{1:1:1 \, |\, 1}$ conditional tripartite informations, but also relationally between the various arguments); in other words, they can be rotated into each other under permutations.   

\paragraph{Another simple example, $\QNknn{6}{5, 5, 1}{12}$:}  
Here we start with a longer expression in the I-basis (though shorter than 
$\QNknn{6}{4,3}{24}$ 
in the S-basis), which has 1 $\sI_5$, 5 $\sI_4$'s, and  5 $\sI_3$'s:
$$
\QNknn{6}{5, 5, 1}{12}
	= - \xI{ABCEF} + \xI{ABCD} + \xI{ABCE} + \xI{ABCF} + \xI{ABEF} + \
\xI{ACEF} - \xI{ABC} - \xI{ABF} - \xI{ACE} - \xI{AEF} - \xI{BCD} 
$$
Since the $\sI_5$ argument is missing a $\x{D}$, the first $\sI_4$ cannot be combined with it naturally, whereas the other 4 can.  Guessing that the first $\sI_4$ and last  $\sI_3$ terms pair (into conditional tripartite information) and the remaining 9 terms group according to \cref{eq:simppol122} (into a polychromatic $\pqty{1:2:2}$ tripartite information), we can readily rewrite this as 
$$
\QNknn{6}{5, 5, 1}{12}
	= - \txI{A:BE:CF} - \ctxI{B:C:D}{A}
$$

\paragraph{Non-uniqueness:}
Permutation symmetry of the arguments of $I_n$'s imply that groupings of polychromatic and/or conditional multipartite informations are non-unique.  For example, if we write $\xI{ABCD}$ using \cref{eq:Inadditivitysplit} in form of 3 $\sI_3$'s, change which two arguments get grouped together, and cancel terms (or equivalently if  we use permutation symmetry in the first group of arguments in $\ctxI{B:C:D}{A} $), we have
\begin{equation} 
\ctxI{B:C:D}{A} 
	= \xI{BCD} - \xI{ABCD} 
	= \txI{AB:C:D} - \txI{A:C:D}
	= \txI{AC:B:D} - \txI{A:B:D}
\end{equation}
which can be re-cast into an identity of the form
\begin{equation} 
\txI{AB:C:D} + \txI{A:B:D} = \txI{AC:B:D} + \txI{A:C:D}
\end{equation}
Note that this is generalizable to any $I_n$, replacing $\x{D}$ by any $\aset{\pI}$ (or even the empty set).

As a more complicated example, from \cref{eq:simpcond3} and \cref{eq:simppol3} we can re-express
\begin{equation} 
\begin{aligned}
\xI{ACD} 
	&= \txI{AB:C:D} - \ctxI{B:C:D}{A} \\
	&= \txI{A:BC:D} - \ctxI{A:B:D}{C} \\
	&= \txI{A:C:BD} - \ctxI{A:B:C}{D} 
\end{aligned}
\end{equation}
where the separate lines interchange which argument the extra party $\x{B}$ not appearing in the LHS gets added on.
A seemingly different kind of identity is apparent in the triplet of compact expressions of \cref{eq:Q18simplif} which were obtained by re-grouping terms. However, all such relations are ultimately rooted in the permutation symmetry of $I_n$.

\paragraph{Signs and term counting:}
Given the multitude of identical expressions, it is useful to have some structural organizing principle which would facilitate determining what kind of a simplification of a given HEI is possible.  
First of all, recall (cf.\ \cref{prop:relcoefbases}) that in any HEI written in the I-basis, the $\sI_n$ terms have coefficient of sign determined by $n$, namely $(-1)^n$.  By \cref{eq:CondMultInf}, we see that any conditional or polychromatic $I_n$ with sign $(-1)^n$ preserves this structure when decomposed in the I-basis.  
This observation suggests that we can characterize each quantity in terms of how many $(-1)^n I_n$ terms it has for each $n$.  For superbalanced HEI quantities, we start at $n=3$, and if we only consider $\N\le 6$, we only need to go up to $n=6$.  However, in all the HEIs we use for examples, the $\sI_6$ term is absent.  So each term is characterized by a triplet of integer coefficients, given by 
$\Ink{i_3,i_4,i_5}$ where $i_n$ is the number of $(-1)^n I_n$ terms.

\begin{table}[htbp]  
\begin{center}
\scriptsize
\begin{tabular}{| c | c | c |}
\hline
IQ &  $i\#$ classification & equation\\ 
\hline
\hline
$-I_{(1:1:1)}$
	& $\Ink{1,0,0}$
	& defn.
\\  \hline   \hline
$I_{(1:1:1:1)}$
	& $\Ink{0,1,0}$
	& defn.
\\  \hline 
$-I_{(1:1:2)}$
	& $\Ink{2,1,0}$
	& \cref{eq:simppol3}
\\  \hline 
$-C_{(1:1:1|1)}$
	& $\Ink{1,1,0}$
	& \cref{eq:simpcond3}
\\  \hline   \hline
$-I_{(1:1:1:1:1)}$
	& $\Ink{0,0,1}$
	& defn.
\\  \hline 
$-I_{(1:1:3)}$
	& $\Ink{3,3,1}$
	& \cref{eq:simppol113}
\\  \hline 
$-I_{(1:2:2)}$
	& $\Ink{4,4,1}$
	& \cref{eq:simppol122}
\\  \hline 
$-C_{(1:1:1|2)}$
	& $\Ink{1,2,1}$
	& \cref{eq:simpcond2}
\\  \hline 
$-C_{(1:1:2|1)}$
	& $\Ink{2,3,1}$
	&\cref{eq:simppolcond}
\\  \hline 
\end{tabular}
\end{center}
\caption{$i\#$ classification for (conditional) multipartite information depending in size of arguments.}
\label{tab:kclassif}
\end{table}

Let us now establish a more refined shorthand for conditional and polychromatic $I_n$'s based on the size of the argument:  $I_{(1:1:1)}$ is the singleton tripartite information while $I_{(1:1:2)}$ has doubleton for one of the arguments, and $C_{(1:1:1|1)}$ etc.\ indicates the corresponding conditional tripartite informations. 
Table \ref{tab:kclassif} shows how these decompose in the $i\#$ classification.
This classification allows us to take a given HEI and immediately extract the various possibilities for writing it out in various polychromatic combinations.  Note that correct $i\#$ decomposition is necessary but not sufficient condition for being able to rewrite a given HEI, since it suppresses the finer structural details.  

In the above two simple examples, we can see the implementation straightforwardly.
$\QNknn{6}{4,3}{24}$ 
can be decomposed as 
$\Ink{4,3,0}= \Ink{2,1,0}+2 \Ink{1,1,0}$, giving the form $-I_{(1:1:2)}-C_{(1:1:1|1)}-C_{(1:1:1|1)}$ (cf.\ the second row block in \cref{tab:HEI_N6}).
Even more simply, $\QNknn{6}{5, 5, 1}{12}$ can be decomposed as 
$\Ink{5,5,1}=\Ink{4,4,1}+\Ink{1,1,0}$, giving the form $-I_{(1:2:2)}-C_{(1:1:1|1)}$.
On the other hand, we can easily find examples where the simplest decomposition doesn't work.  Below we examine one such case in full detail.

\paragraph{Cyclic inequality and purification:}
Consider $\QNknn{5}{3, 2}{5}$ in the first row block of \cref{tab:HEI_N5}.  In this form, 
$$ \QNknn{5}{3, 2}{5}=  \xI{ABCD} + \xI{ABCE} - \xI{ABC} - \xI{ABE} - \xI{ACD} $$
has the suggestive decomposition as $\Ink{3,2,0}=\Ink{2,1,0}+\Ink{1,1,0}$, namely as 
$-I_{(1:1:2)}-C_{(1:1:1|1)}$, which can in fact be achieved in two different ways (by alternate groupings), giving
\begin{equation} 
\begin{aligned}
\QNknn{5}{3, 2}{5}
	&=  -\txI{A:B:CE} - \ctxI{A:C:D}{B}  \\
	&= -\txI{A:C:BD} - \ctxI{A:B:E}{C}
\end{aligned}
\end{equation}

If however we use the corresponding S-basis expression,
$$
\QNknn{5}{3, 2}{5}= - \xS{ABCD} - \xS{ABCE} - \xS{BC} - \xS{BD} - \xS{CE} - \xS{A} + \xS{ABC} + \xS{ABD} + \xS{ACE} + \xS{BCD} + \xS{BCE}
$$
and purify on $\x{A}$, under additionally replacing $\x{\{E,O,D,B,C\}} \to \x{\{A,B,C,D,E\}}$, we obtain the cyclic (dihedral) HEI in the canonical form,
$$
\QNknn{5}{5,5,1}{5}= - \xS{ABCDE} - \xS{AB}- \xS{BC} - \xS{CD} - \xS{DE} - \xS{EA} + \xS{ABC} + \xS{BCD}+ \xS{CDE} + \xS{DEA}+ \xS{EAB}
$$
which now has a longer form in the I-basis,
$$
\QNknn{5}{5,5,1}{5} =
- \xI{ABCDE}  + \xI{ABCD} + \xI{ABCE} + \xI{ABDE} + \xI{ACDE} + \xI{BCDE} -\xI{ABD} - \xI{ACD} - \xI{ACE} - \xI{BCE} - \xI{BDE}
$$
This demonstrates the important fact that the $i\#$ classification is not purification-invariant (though it is of course invariant under the smaller color permutation group; cf.\ \cref{fn:inumberclass}).

Let us now try to obtain a compact form for this expression.
One natural decomposition that suggests itself is 
$\Ink{5,5,1}=\Ink{4,4,1}+\Ink{1,1,0}$, giving the form $-I_{(1:2:2)}-C_{(1:1:1|1)}$ (which is what we utilized for $\QNknn{6}{5, 5, 1}{12}$).  However, the specific form of the present $\QNknn{5}{5,5,1}{5}$ does not allow such a decomposition.  This is because the $-I_{(1:2:2)}$ term requires all 4 $\sI_3$'s to have a common argument, whereas in $\QNknn{5}{5,5,1}{5}$ only groups of 3 $\sI_3$'s have a common argument.  To get around this, we could add and subtract a spurious $\sI_3$ so as to complete the requisite form.  This would give
$\QNknn{5}{5,5,1}{5} = -\txI{BC:DE:A} - \txI{CD:B:E} + \txI{A:B:E}$, which is still compact, but has the wrong sign, and therefore is not of the \ctform: Its $i\#$ decomposition is $\Ink{5,5,1}=\Ink{4,4,1}+\Ink{2,1,0}-\Ink{1,0,0}$.  
Similarly, if we try the potential decomposition
$\Ink{5,5,1}=\Ink{3,3,1}+\ldots$, we find that this likewise does not have the structurally necessary form, which requires all 3 of the $\sI_3$'s to have two arguments in common.

A more complicated possibility is to use the decomposition
$\Ink{5,5,1}=\Ink{1,2,1}+\Ink{2,1,0}+2\Ink{1,1,0}$, which now does admit the corresponding \ctform; in particular, 
\begin{equation} 
\QNknn{5}{5,5,1}{5} = -\ctxI{A:B:D}{CE} - \txI{AB:C:E} -\ctxI{A:C:D}{E} -\ctxI{B:D:E}{C}
\end{equation}
 This now has correct signs, but at the expense of a longer expression.  

 In fact, noting that   $\Ink{5,5,1}=\Ink{1,1,0}+\Ink{2,3,1}+\Ink{2,1,0}$, we might seek a combination of the form $-C_{(1:1:1|1)}-C_{(1:1:2|1)}-I_{(1:1:2)}$, and we leave it to the reader as an exercise to check that indeed a suitable regrouping gives the requisite form,\footnote{ \, 
 In fact, originally this was obtained by purifying the tripartite form of $\QNknn{5}{3, 2}{5}$ and using the relevant identities.  
 }
\begin{equation} 
\QNknn{5}{5,5,1}{5} = -\ctxI{A:C:D}{B} -\ctxI{AB:C:E}{D} - \txI{AE:B:D} 
\end{equation}

 \paragraph{Transmutation of $I_3$'s under purifications:}
Motivated by the previous example, it is clear that to extract the full power of groupings, we need to understand how the  $I_3$'s transform under purifications, which will allow us to rewrite and recombine the compact expressions in various ways.
While we leave the full exploration for future work,
two useful observations pertaining to the tripartite information, which follow as a direct consequence of \cref{lem:Ifullsystem}, are the following.  
For any (possibly composite) subsystems $\x{X},\x{Y},\ldots$, the tripartite information for any tripartition of the full system vanishes,
\begin{equation}\label{eq:I3pur0}
	I_3(\x{X:Y:}\cplmt{\x{XY}}) = 0 \ ,
\end{equation}
while for any quadripartition, we can swap any of the arguments for the missing one,
\begin{equation}\label{eq:I3swap}
	I_3(\x{X:Y:Z}) = I_3(\x{X:Y:}\cplmt{\x{XYZ}}) \ .
\end{equation}

Furthermore, we can use this to build up the analogous relation for conditional tripartite information:
\begin{equation}\label{eq:CI3swap}
	I_3(\x{X:Y:Z} \, | \, \x{W}) = - I_3(\x{X:Y:Z}  \, | \,  \cplmt{\x{XYZW}})
\end{equation}
which follows by expanding out the LHS in terms of polychromatic $I_3$'s, e.g.
\begin{equation}\label{eq:}
	I_3(\x{X:Y:Z} \, | \, \x{W}) = I_3(\x{X:Y:ZW}) - I_3(\x{X:Y:W}) 
\end{equation}
applying \cref{eq:I3swap} on each term, and finally recombining to get the RHS of \cref{eq:CI3swap}.  
Notice  however that we have now broken the automatic $I_n$ sign alternation.  In particular, when the purifier appears explicitly in one of the arguments of a conditional tripartite information, its I-basis expansion will not have each $\sI_n$ term with the sign given by $(-1)^n$.
Hence if the purifier were to appear in some conditional $I_3$ in the expression for a HEI, it must be the case by \cref{prop:relcoefbases} that 
any `wrong-sign' $\sI_n$ terms in the I-basis are cancelled out by other terms in the HEI.

For a more involved example of seeing this in operation, consider e.g.\ 
$\QNknn{6}{12, 14, 4}{1716} $, 
\begin{equation}
-\tI{A:BD:CF}-\tI{AE:BO:C}-\ctI{A:B:DE}{C}-\ctI{BF:D:O}{A}
\end{equation}
This involves the purifier $\x{O}$, so it's not in the \ctform\ yet.  If we convert this expression directly to the I-basis and then collect terms, we find that we can indeed recast this into \ctform, but a 6-term (IICCCC) one:
\begin{equation}
-\tI{A:B:F}-\tI{AE:C:DF}-\ctI{B:C:D}{A}-\ctI{A:BD:C}{F}-\ctI{AD:B:E}{C}-\ctI{CE:D:F}{B}
\end{equation}
However, if we first permute $\x{O}$ with $\x{F}$, we can directly regroup terms, and obtain a shorter \ctform\ IICCC,
\begin{equation}
-\tI{CE:D:F}-\tI{AE:BF:C}-\ctI{A:B:DE}{C}-\ctI{A:D:E}{F}-\ctI{A:B:EF}{D}
\end{equation}
This illustrates that, even if we have a given HEI in the \ctform, there might in principle be a purification which can yield an even more compact \ctform\ expression.
Our search strategy guarantees that the \ctform\ we have found for all HEIs involves the minimal number of terms possible.

\section{Systematic search for novel HEIs}
\label{asec:n6heis}

The strategy for generating novel HEIs systematically, by exploiting the structural rigidity of our \ctform, is as follows. We first pick a \ctform\ structure by fixing the number of $I_3$'s and conditional $I_3$'s to be combined. For instance, we may consider expressions of the form
\begin{equation}
\label{eq:egiic}
    - I_3(\x{X}_1:\x{Y}_1:\x{Z}_1) - I_3(\x{X}_2:\x{Y}_2:\x{Z}_2) - I_3(\x{X}_3:\x{Y}_3:\x{Z}_3 \, | \, \x{W}_3).
\end{equation}
To be exhaustive on such \ctform\ structures, we consider all possible disjoint subsystems in $[\nN]$ for each of the terms. In the case of $I_3$, these can be accounted for by the collection of all partitions of $[\nN+1]$ into four nonempty subsets, where the three of them not involving the purifier are chosen as arguments for $I_3$. Hence all choices of $I_3$ can be enumerated by Stirling numbers of the second kind $\left\{ {\genfrac{}{}{0pt}{}{\nN+1}{4}} \right\}$ (e.g.\ $350$ for $\nN=6$), where\footnote{ \, Partitions involving empty subsets can be ignored as they would give vanishing $I_3$ by \cref{lem:Ifullsystem}.}
\begin{equation}
    \left\{ {\genfrac{}{}{0pt}{}{n}{k}} \right\} = \sum_{i=0}^k \frac{(-1)^{k-i}\, i^n}{(k-i)!\, i!}
\end{equation}
The logic for conditional $I_3$ is similar: in this case we consider the collection of all partitions of $[\nN+1]$ into five nonempty subsets, where the four of them not involving the purifier are chosen as arguments. In this case the fourth argument of conditional $I_3$ is not symmetric with respect to the others, so every partition into five subsets ends up giving four distinct expressions. The resulting total number of choices is thus given by $ 4\times\left\{ {\genfrac{}{}{0pt}{}{\nN+1}{5}} \right\}$ (e.g.\ $560$ for $\nN=6$).

We have exhaustively searched for HEIs captured by all expressions involving up to four terms in \ctform.\footnote{ \, In fact, we have performed an even more comprehensive exercise. To better assess the effectiveness of the \ctform, we have executed this exhaustive search with up to four terms relaxing \cref{def:ctform} to allow for subsets containing the purifier. This generically results in expressions which violate the sign-alternation required by \cref{prop:relcoefbases}, but may plausibly yield valid HEIs when wrong-signed terms end up canceling out. In such cases, it would be logically plausible for the resulting HEIs to not admit a recasting in actual \ctform\ where the purifier is disallowed. Our results show that, despite the largely enhanced combinatorics that the inclusion of the purifier causes, allowing for the purifier to appear in \ctform\ expressions of up to four terms only results in an additional $18$ HEIs (out of over $300$), all of which do in fact admit actual \ctform\ expressions without the purifier by merely allowing for more than four terms. We thus conclude that the \ctform\ in its restricted version of \cref{def:ctform} motivated by the sign-alternation property is robust and effective, and has the potential of being able to capture all HEIs.}
To be more explicit, we will use the shorthand introduced in \cref{ss:HEIrecasting}, of referring to $I_3$ and conditional $I_3$ schematically by $\x{I}$ and $\x{C}$, respectively. Then an example of \ctform\ we have exhausted is $\x{IIC}$, which consists of all expressions like \cref{eq:egiic}. The process of exhausting a given structure entails, for all resulting inequalities, proving them valid (via the contraction map proof method) or invalid (via explicit violation by a holographic entropy vector) and, for all valid ones, proving them facets (via a rank test against holographic entropy vectors) or redundant (via a redundancy test against obtained facets).
We have carried out such a comprehensive analysis, but since the whole process is highly nontrivial in practice, here we will just summarize the basic strategy.

The process is computationally heavy from the onset: a brute-force generation of all \ctform\ expressions e.g.\ of the form $\x{IICC}$ would yield thousands of millions of candidate inequalities, which is too expensive memory-wise. Filtering by permutation and purification symmetries at multiple stages in the combinatorial combination of terms in all possible ways is necessary already at the level of generating the data. This reduces the size of the starting dataset by a factor of $(\nN+1)!$, the order of $Sym_{\nN+1}$. Once the data is available, many candidate inequalities can be easily rejected by testing them against simple $\nN=6$ holographic entropy vectors. For example, for $\x{IICC}$ this reduces the set of candidates inequalities to just about a thousand orbits. 
One can then attempt to prove those which are not violated using the contraction map proof method and, when successful, perform a facetness check on proven inequalities using holographic entropy vectors for $\nN=6$. Those valid inequalities which are proven to be facets can be definitively declared HEIs. However, the latter two steps are not necessarily conclusive if negative, and thus a variety of additional cross-checks is necessary to be genuinely exhaustive. For checking validity, the contraction map proof method is highly computationally inefficient, which means it will often time out after a few days with an inconclusive output. In such cases, the corresponding inequalities cannot be declared valid nor invalid. Similarly, the facetness check is also inconclusive if false, as it relies on a data set of holographic entropy vectors for $\nN=6$ obtained semi-randomly, rather than of the complete set of extreme rays (which is unknown).\footnote{ \, If the validity test is positive but the facetness test is inconclusive, then we know the inequality is valid but it remains unknown whether it is a facet or redundant with respect to actual facets. If the validity test is inconclusive but the facetness test is positive, then we know that if valid the inequality would define a facet, but it remains unknown whether it actually is valid. If both the validity and facetness tests are inconclusive, then we know nothing.}
To bypass these difficulties, every new HEI found can be used to improve on the inconclusive results above. Inequalities left inconclusive by the contraction map proof and/or by the facetness check can be tested for redundancy against all found HEIs, and if redundant, conclusively declared valid inequalities which are redundant. If the redundancy test is negative though, the status of the pertinent inequality would remain as inconclusive as before.
At the end of the day, we declare a particular \ctform\ exhausted if we can assign a conclusive status to every inequality of the given form, namely, if the following is true: all inequalities have been proven invalid or valid, and all the latter have been proven either facets or redundant.

Up to small set of valid inequalities whose facet status we have not been able to establish, we have successfully exhausted all expressions with the structures $\x{I}$, $\x{IC}$, $\x{IIC}$, $\x{ICC}$, $\x{IIIC}$, $\x{IICC}$, and $\x{ICCC}$.\footnote{ \, Structures involving more than one $\x{I}$ but no $\x{C}$'s are always redundant, as $\x{I}$ itself is a facet, while structures involving only $\x{C}$'s turn out to always give invalid inequalities.} As listed, we are ordering structures by simplicity: first by number of terms, and then by number of $\x{C}$'s (which are structurally more complicated than the $\x{I}$'s). Oftentimes a given HEI can be captured by more than one structure. We will say a given structure generates a certain HEI if it is the simplest structure capturing that HEI. With this in mind, we can now state our quantitative findings in the first 7 columns of \cref{tab:ICformnums}. 



\begin{table}
    \centering
    \footnotesize
    \begin{tabular}{|c||c|c|c|c|c|c|c|c|c|c|c||c|}
        \hline 
        Form  & I & IC & IIC & ICC & IIIC & IICC & ICCC & IIICC & IICCC & I${^2}$C${^4}$ & I${^2}$C${^5}$ & total \\
        \hline \hline
        facets & 3 & 3 & 13 & 3 & 92 & 197 & 1 & 969 & 447  & 144 & 4 & 1876 \\
        \hline 
        lifts & 3 & 2 & 2 & 1 & 0 & 0 & 1 & 0 & 0 &  1 & 0 &  10  \\         \hline
    \end{tabular}{}
    \caption{The number of HEIs with simplest \ctform\ of the type specified.  (This excludes the single SA, which is not superbalanced and therefore not in the \ctform.)
    Up to 4 terms in the \ctform, the search is mostly exhaustive; only for IIIC and IICC forms there are an additional 7 and 5 valid inequalities, respectively, whose facet status remains inconclusive. 
    The last row indicates how many of the enumerated facets are just lifts from $\nN=5$ HEIs.
    }
    \label{tab:ICformnums}
\end{table}

Besides the exhaustive analysis described above, we have also pursued another systematic approach that has turned out to be extremely fruitful, yielding more than 1800 new HEIs at the time of writing, summarized in the remaining columns of \cref{tab:ICformnums}. It relies on two simple observations. The first one is that, if some inequality is valid, then the addition of any instance of $I_3$ will always make it redundant. This statement is straightforward to prove. The second one is that, if some inequality is not valid, then the addition of any instance of conditional $I_3$ never seems to make it valid. This is a statement we are not proving, but for which we have encountered no counterexamples. These two observations can be combined into the following systematic strategy. Starting from any set of inequalities in \ctform, separate them into valid and not valid. Then we know that adding instances of $I_3$ to the valid ones will never generate a facet, and that adding instances of conditional $I_3$ to the not valid ones will not generate valid inequalities. Hence we restrict the generation of data to only adding instances of $I_3$ to inequalities that are not valid, and instances of conditional $I_3$ to inequalities which already are valid. The former has the potential to turn inequalities which are not valid into inequalities which are valid, while the latter has the potential to turn inequalities which are valid (but possibly redundant) into inequalities which are valid and facets. In other words, the general heuristic is that $I_3$ and conditional $I_3$ respectively point in an outward and inward direction in entropy space, and therefore the former favors validity, whereas the latter favors tightness.

Although we have by no means pursued this search to its fullest extend, we believe the finding of more than 1800 new HEIs through this strategy definitely motivates further exploration. Given that the search has not been comprehensive on this front, we simply report on a subset of them.
\footnote{ \, These also appear in \cite{hecdata}, whose intent is to have a more up-to-date repository of the HEC data.}
Our ancillary files contain the following data:

\begin{itemize}
    \item \texttt{HEIvectors.txt}: List of representative HEIs for each of the $1877$ distinct facet orbits currently known for the $\nN=6$ HEC. These are given one per row by their coefficients in the entropy basis following the canonical order for subsystems (first by cardinality, then lexicographically).
    The first $11$ rows correspond to the uplifts of $\nN\le5$ HEIs to $\nN=6$. 
    The rest are all novel $\nN=6$ HEIs found exhaustively using \ctform\ expressions involving up to four terms and also via the systematic search explained above. 
    \item \texttt{HEIforms.txt}: List of \ctform\ expressions for all HEIs given in \texttt{HEIvectors.txt}, following the same order.\footnote{\, The first HEI is SA, which obviously has no \ctform.} These are given one per row by 
    a list of either $3$ or $4$ terms with the arguments of $\x{I}$ and $\x{C}$ terms, respectively. In the case of $\x{C}$ terms, the forth argument corresponds to the subsystem conditioned on. Within every cluster of rows specified in the ordering of \texttt{HEIvectors.txt} (i.e., rows $1-11$ for lifts, and the rest), notice that HEIs have been ordered by increasing complexity of their \ctform\ expressions.
\end{itemize}

To give the reader some feel for the richness of the HEI forms which can appear, we characterize them by two salient quantities.  The first is the $i\#$ classification utilized above (cf.\ \cref{tab:kclassif} for the $i\#$ classification for the various I and C terms), presented as the quadruplet $\{\# \sI_3,\# \sI_4,\# \sI_5,\# \sI_6\}$ counting the corresponding number of singleton $\sI_n$'s in the I-basis.  The second, which we will label by $m$, counts the total number of arguments when written out in the \ctform\ (so for instance a singleton $\sI_3$ would have $m=3$), which provides a distinct measure of complexity compared to the $i\#$ classification.  
To consolidate the above data, we tabulate the frequency (i.e.\ number of distinct HEI orbits) of a given $i\#$ classification and $m$-number in a 2-d array
in \cref{fig:HEIimfreq}, with the background color-coded by the corresponding I..C.. form (which is unique for the given $\{i\#,m\}$ pair).  For example, the most frequent occurrence comprises of 132 distinct orbits which all have the form IIICC, $i\#$ classification $\Ink{12,12,3,0}$, and $m=23$.

\begin{figure}[htbp] 
\begin{center}
\includegraphics[width=6in]{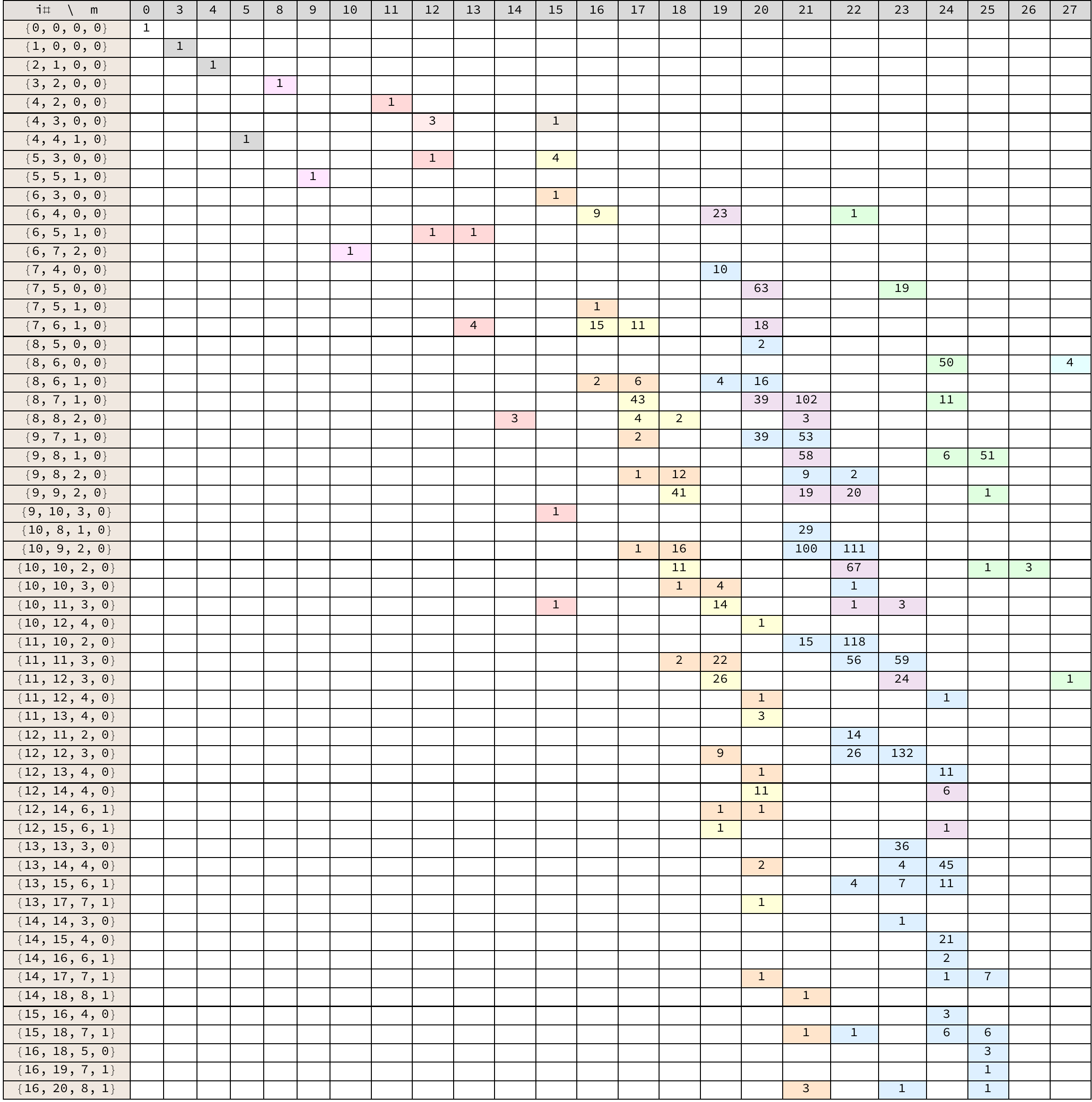}
\caption{The number of HEI orbits hitherto found (at the time of submitting v2) with the given i\# classification (rows) and $m$-number (columns).  Blank spaces correspond to zero frequency, and the shading indicates the I..C.. form: gray=I, magenta=IC, red=IIC, pink=ICC, orange=IIIC, yellow=IICC, brown=ICCC, blue=IIICC, purple=IICCC, green=IICCCC, and cyan=IICCCCC.  The total of all numbers is 1877.
}
\label{fig:HEIimfreq}
\end{center}
\end{figure} 

While up to 4 terms in the \ctform\ the HEI search was mostly exhaustive, for higher number of terms the search is not yet complete.  Nevertheless, we see many structural similarities.  Since in v1 we focused largely on the exhaustive part, for the reader's convenience, we devote \cref{tab:new_HEIs} to a representative set of new HEI orbits found in v2, which all have at least 5 terms.

\begin{table}
   \begin{center}
   \tiny
   \setlength\tabcolsep{5pt}
    \begin{tabular}{|c||c|c|c|c|}
        \hline 
        $n$ & IC form & $i\#$ & $m$ & \ctform \\
            \hline 
        331&IIICC&\{8, 6, 1, 0\}&20&$-\tI{A:B:CDF}-\tI{A:D:E}-\tI{B:C:DE}-\ctI{A:C:F}{E}-\ctI{A:E:F}{B}$\\
         347&IIICC&\{9, 7, 1, 0\}&20&$-\tI{A:B:EF}-\tI{A:BD:CE}-\tI{B:C:D}-\ctI{A:C:F}{D}-\ctI{A:E:F}{B}$\\
         386&IIICC&\{9, 7, 1, 0\}&21&$-\tI{A:B:CD}-\tI{A:BE:C}-\tI{A:D:EF}-\ctI{B:C:F}{D}-\ctI{B:CE:D}{A}$\\
         439&IIICC&\{10, 8, 1, 0\}&21&$-\tI{A:B:CE}-\tI{A:B:DF}-\tI{A:CD:EF}-\ctI{A:C:E}{B}-\ctI{A:D:F}{B}$\\
         477&IIICC&\{10, 9, 2, 0\}&21&$-\tI{A:B:C}-\tI{A:BF:D}-\tI{AC:DF:E}-\ctI{A:B:D}{E}-\ctI{B:DF:E}{A}$\\
         577&IIICC&\{11, 10, 2, 0\}&21&$-\tI{A:B:F}-\tI{A:BE:CD}-\tI{AD:BF:C}-\ctI{A:C:E}{F}-\ctI{A:E:F}{C}$\\
         594&IIICC&\{10, 9, 2, 0\}&22&$-\tI{A:B:CD}-\tI{A:B:DF}-\tI{ABD:C:F}-\ctI{A:BF:E}{C}-\ctI{C:D:E}{F}$\\
         705&IIICC&\{11, 10, 2, 0\}&22&$-\tI{A:B:CD}-\tI{A:BC:DF}-\tI{A:DF:E}-\ctI{BF:C:E}{A}-\ctI{C:D:F}{E}$\\
         838&IIICC&\{11, 11, 3, 0\}&22&$-\tI{A:B:C}-\tI{A:BC:DE}-\tI{AC:B:DF}-\ctI{C:E:F}{AB}-\ctI{D:E:F}{C}$\\
         894&IIICC&\{12, 12, 3, 0\}&22&$-\tI{A:B:C}-\tI{A:BE:CD}-\tI{AB:D:EF}-\ctI{A:B:F}{D}-\ctI{C:DF:E}{A}$\\
         925&IIICC&\{11, 11, 3, 0\}&23&$-\tI{A:B:CDF}-\tI{A:C:EF}-\tI{BE:D:F}-\ctI{A:CF:E}{D}-\ctI{AB:C:D}{F}$\\
         984&IIICC&\{12, 12, 3, 0\}&23&$-\tI{A:B:CD}-\tI{A:B:DF}-\tI{AC:B:DE}-\ctI{A:CF:E}{B}-\ctI{AF:C:D}{B}$\\
        1116&IIICC&\{13, 13, 3, 0\}&23&$-\tI{A:B:CE}-\tI{A:BC:DE}-\tI{AC:B:EF}-\ctI{AF:C:D}{B}-\ctI{C:D:E}{A}$\\
        1177&IIICC&\{13, 14, 4, 0\}&24&$-\tI{A:B:CD}-\tI{A:BC:DE}-\tI{AF:BC:D}-\ctI{AE:B:F}{D}-\ctI{B:C:E}{AF}$\\
         1222&IIICC&\{14, 15, 4, 0\}&24&$-\tI{A:B:CD}-\tI{AB:C:DF}-\tI{AB:CE:F}-\ctI{A:BE:F}{C}-\ctI{AB:D:E}{F}$\\
         1246&IIICC&\{13, 15, 6, 1\}&24&$-\tI{A:B:CD}-\tI{A:BF:DE}-\tI{ABC:D:E}-\ctI{A:BCE:F}{D}-\ctI{A:C:F}{E}$\\
         1284&IICCC&\{6, 4, 0, 0\}&19&$-\tI{A:B:D}-\tI{AC:D:E}-\ctI{A:C:F}{D}-\ctI{B:C:E}{D}-\ctI{C:D:F}{E}$\\
         1307&IICCC&\{7, 5, 0, 0\}&20&$-\tI{A:B:CD}-\tI{AB:E:F}-\ctI{B:C:D}{E}-\ctI{C:D:E}{F}-\ctI{C:D:F}{B}$\\
         1370&IICCC&\{7, 6, 1, 0\}&20&$-\tI{A:B:C}-\tI{ABD:C:F}-\ctI{A:D:E}{C}-\ctI{B:D:E}{C}-\ctI{C:D:E}{F}$\\
         1388&IICCC&\{8, 7, 1, 0\}&20&$-\tI{A:B:F}-\tI{AF:BC:D}-\ctI{A:C:E}{B}-\ctI{B:C:E}{D}-\ctI{B:D:E}{F}$\\
         1428&IICCC&\{8, 7, 1, 0\}&21&$-\tI{A:B:CE}-\tI{A:BC:E}-\ctI{A:CD:F}{E}-\ctI{A:E:F}{B}-\ctI{B:C:D}{A}$\\
         1530&IICCC&\{9, 8, 1, 0\}&21&$-\tI{A:B:CE}-\tI{A:BF:DE}-\ctI{A:C:D}{B}-\ctI{B:C:D}{A}-\ctI{B:C:E}{A}$\\
         1817&IICCCC&\{9, 8, 1, 0\}&25&$-\tI{A:B:CE}-\tI{A:B:CF}-\ctI{A:C:F}{B}-\ctI{A:DF:E}{B}-\ctI{B:C:D}{A}-\ctI{B:C:D}{A}$\\
         1732&IICCCC&\{7, 5, 0, 0\}&23&$-\tI{A:B:EF}-\tI{C:E:F}-\ctI{A:C:D}{B}-\ctI{B:C:D}{F}-\ctI{B:D:E}{C}-\ctI{B:D:F}{E}$\\
         1800&IICCCC&\{8, 7, 1, 0\}&24&$-\tI{A:B:E}-\tI{A:BD:C}-\ctI{A:C:D}{E}-\ctI{B:C:EF}{A}-\ctI{B:D:F}{C}-\ctI{C:D:E}{B}$\\
         \hline
    \end{tabular}{}
    \caption{Representative HEIs with $\ge5$ terms in the \ctform, selected so that their $\{i\#,m\}$ pairs were absent in v1, but in the v2 set there are at least 10 such HEIs with that given $\{i\#,m\}$.}
    \label{tab:new_HEIs}
\end{center}
\end{table}

\newpage
\addcontentsline{toc}{section}{References}
\bibliographystyle{JHEP}
\bibliography{references.bib}

\end{document}